\documentclass[10pt]{extarticle}
\usepackage{graphicx} 

\usepackage[utf8]{inputenc}

\usepackage{amsmath,amsfonts,amsthm}

\usepackage{amssymb}

\usepackage{mathtools}

\usepackage{thmtools}
\usepackage{thm-restate}

\usepackage{bm}
\usepackage{bbm}

\usepackage{stackengine}

\usepackage{xcolor}

\usepackage{tikz}
\usepackage{pgfplots}
\pgfplotsset{compat=1.18}

\usepackage{pgfplotstable}

\usetikzlibrary{arrows,arrows.meta}
\usetikzlibrary{positioning,fit}
\usetikzlibrary{patterns}
\usetikzlibrary{decorations.pathmorphing,decorations.text,shapes.geometric}
\usetikzlibrary{fadings}
\usetikzlibrary{spy,backgrounds}

\usepgfplotslibrary{fillbetween}

\usepackage{nicefrac}

\usepackage{multicol}

\usepackage{relsize}

\usepackage[inline]{enumitem}

\usepackage{upgreek}

\usepackage{yfonts}
\usepackage{pifont}
\usepackage[integrals]{wasysym}

\usepackage{nicematrix}

\usepackage{multirow}

\usepackage{subcaption}

\usepackage{url}

\usepackage{hyperref} 

\usepackage[noabbrev]{cleveref}

\usepackage[leftcaption]{sidecap}

\usepackage{microtype}

\usepackage{algorithm}
\usepackage{algpseudocode}

\usepackage{placeins}

\usepackage{appendix}

\DeclareMathOperator*{\argmin}{argmin}
\DeclareMathOperator*{\argmax}{argmax}

\DeclareMathOperator{\Range}{Range}

\DeclareMathOperator{\Diam}{Diam}

\newcommand{\grad}{\nabla}

\DeclarePairedDelimiter{\ceil}{\lceil}{\rceil}

\newcommand{\1}{\mathbbm{1}}

\newcommand{\R}{\mathbb{R}}
\newcommand{\RNN}{\R_{\mathsmaller{0+}}}

\DeclareMathOperator{\sgn}{sgn}

\DeclareMathOperator*{\Expect}{\mathbb{E}}

\DeclareMathOperator*{\Prob}{\mathbb{P}}

\newcommand{\x}{\vec{x}}

\newcommand{\lv}{\mathcal{S}}
\newcommand{\epsv}{\bm{\varepsilon}}
\newcommand{\wv}{\bm{w}}

\newcommand{\distributed}{\thicksim}

\newcommand{\Mean}{\mathrm{M}}

\newcommand{\ProbDist}{\mathcal{D}}

\newcommand{\frange}{r}

\DeclareFontShape{OMX}{cmex}{m}{b}{<-> cmexb10}{}
\SetSymbolFont{largesymbols}{bold}{OMX}{cmex}{m}{b}

\newcommand{\todo}[1]{\textcolor{red!50!black}{ToDo: #1}}

\newcommand{\cyrus}[1]{\textcolor{green!50!black}{Cyrus: #1}}

\newcommand{\Input}{\textbf{input}: }
\newcommand{\Output}{\textbf{output}: }

\newcommand{\wrt}{w.r.t.}

\newcommand{\NGroups}{g}

\makeatletter
\providecommand{\leftsquigarrow}{%
  \mathrel{\mathpalette\reflect@squig\relax}%
}
\newcommand{\reflect@squig}[2]{%
  \reflectbox{$\m@th#1\rightsquigarrow$}%
}
\makeatother

\DeclarePairedDelimiter{\abs}{\lvert}{\rvert}
\DeclarePairedDelimiter{\norm}{\lVert}{\rVert}

\tikzset{
  declare function={
      sgn(\x) = (and(\x<0, 1) * -1) + (and(\x>0, 1) * 1);
      gm2(\va,\vb) = (sqrt(\va * \vb ));
      pm2(\va,\vb,\p) = ((\va ^ \p + \vb ^ \p) / 2) ^ (1 / \p);
      pmw2(\va,\vb,\wa,\wb,\p) = (\wa * \va ^ \p + \wb * \vb ^ \p) ^ (1 / \p);
      amean3(\va,\vb,\vc) = ((\va + \vb + \vc) / 3);
      as3(\va,\vb,\vc,\pp) = ((\va ^ \pp + \vb ^ \pp + \vc ^ \pp) / 3);
      pm3(\va,\vb,\vc,\pp) = as3(\va,\vb,\vc,\pp) ^ (1 / \pp);
      asw3(\va,\vb,\vc,\wa,\wb,\wc,\pp) = (\wa * \va ^ \pp + \wb * \vb ^ \pp + \wc * \vc ^ \pp);
      pmw3(\va,\vb,\vc,\wa,\wb,\wc,\pp) = asw3(\va,\vb,\vc,\wa,\wb,\wc,\pp) ^ (1 / \pp);
    },
}

\usepackage{tikz-3dplot}

\usepackage{xcolor}

\usepackage{natbib}

\usepackage{authblk}

\usepackage{thmtools}
\usepackage{thm-restate}

\theoremstyle{definition}

\newtheorem{definition}{Definition}[section]

\newtheorem{lemma}[definition]{Lemma}

\newtheorem{theorem}[definition]{Theorem}

\newtheorem{example}[definition]{Example}

\title{{Algorithms and Analysis for Optimizing Robust Objectives \protect\\ in Fair Machine Learning}}

\author{Cyrus Cousins}
\affil{University of Massachusetts Amherst
\\ Columbia Workshop on Fairness in Operations and AI}
\date{December 2023}

\makeatletter
\hypersetup{pdftitle={\@title},pdfauthor={Cyrus Cousins},pdfsubject={Adversarial Fair Learning},pdfkeywords={Fair Machine Learning, Rawlsian Ethics, Adversarial Learning, Convex Optimization, Robust Fair Learning}}
\makeatother

\newcommand{\Wv}{\mathcal{W}}
\newcommand{\Lv}{\mathcal{S}}
\newcommand{\M}{\mathcal{M}}

\renewcommand{\lv}{\bm{s}}

\DeclareMathOperator*{\minmax}{\substack{\min \\ \max}}
\DeclareMathOperator*{\argminmax}{arg\minmax}

\DeclareMathOperator*{\maxmin}{\substack{\max \\ \min}}
\DeclareMathOperator*{\argmaxmin}{arg\maxmin}

\DeclareMathOperator*{\infsup}{\substack{\inf \\ \sup}}
\DeclareMathOperator*{\supinf}{\substack{\sup \\ \inf}}

\DeclareMathOperator{\CH}{CH}

\DeclareMathOperator{\Payoff}{P}

\newcommand{\titlefont}{\it\bfseries}

\colorlet{titlecol}{black!80!red}

\usepackage{blkarray}

\setlist{wide,labelwidth=0pt,labelindent=0pt,labelsep=3pt,nosep,topsep=0pt}
\setlist[enumerate,1]{label=\arabic*),ref=\arabic*}
\setlist[enumerate,2]{label=\Alph*),ref=\Alph*}

\usepackage{geometry}
\newif\ifextensions
\extensionsfalse

\newif\ifsoft
\softtrue

\newif\ifforc
\forcfalse

\newif\ifdraft
\draftfalse

\extensionstrue

\newcommand{\Dae}{\mathrm{D\text{\ae}}}
\newcommand{\Ang}{\mathrm{Ang}}

\newcommand{\question}[1]{\textcolor{purple}{\textbf{Question:} \emph{#1}}}

\iftrue
\renewcommand{\todo}[1]{}
\renewcommand{\cyrus}[1]{}
\renewcommand{\question}[1]{}
\fi

\begin{document}

\maketitle

\begin{abstract}

The \emph{original position} or \emph{veil of ignorance} argument of John Rawls, 
perhaps the most famous argument for egalitarianism, 
states
that our concept of fairness, justice, or welfare should be decided from behind a veil of ignorance, and thus
must
consider everyone impartially (invariant to our identity). 
This can be posed as a
zero-sum game, where a D\ae mon constructs a world, and an adversarial Angel then places the D\ae mon
into the
world.
This
game
incentivizes the D\ae mon to maximize the minimum utility over all people (i.e., to maximize \emph{egalitarian welfare}).
In some sense, this 
is the most extreme
form of \emph{risk aversion} or \emph{robustness}, and we show that by weakening the
Angel, milder
robust objectives arise, which we argue are effective \emph{robust proxies} for \emph{fair} learning or allocation tasks.
In particular, the utilitarian, Gini, and power-mean welfare concepts
arise from special cases of the
adversarial game, which has philosophical implications for the understanding of each of these concepts.
We also motivate a new fairness concept that essentially fuses the nonlinearity of the power-mean with the piecewise nature of the Gini class.
\ifextensions
Then, exploiting the 
relationship
between fairness and robustness,
we show that these robust fairness concepts can all be efficiently optimized under mild conditions via standard maximin optimization techniques.
Finally, we show that such methods apply in machine learning contexts, and moreover we show generalization bounds for robust fair machine learning tasks.
\fi

\bigskip

\centering
\textbf{Keywords:} \\
\medskip
\it
Fair Machine Learning | Rawlsian Ethics | Adversarial Learning | Convex Optimization | Robust Fair Learning
\end{abstract}

\section{Introduction}

\todo{Terrible intro paragraph!}

Fairness and robustness are crucial aspects of machine learning and allocation systems, both of which are generally addressed through modelling, data collection, and objective selection.
This work extends
ideas and objectives in welfare-centric fair machine learning and optimization
introduced by \citet{cousins2021axiomatic,cousins2021bounds,cousins2022uncertainty,cousins2023revisiting}.
We derive robust variants of fair objectives, and explore mathematical and philosophical connections between robustness and fairness.
In particular, we consider robust \emph{welfare functions}, which aggregate \emph{utility} across a population, and robust \emph{malfare functions}, which aggregate \emph{disutility}, both serving as fairness metrics and as optimization targets.
We then combine these robust objectives with adversarial optimization theory and techniques, which expands on the relationship between fairness, robustness, and uncertainty in machine learning and allocation problems \citep{mazzetto2021adversarial,dong2022decentering,cousins2023into
}.

The core of this paper is the construction of a hierarchy of Rawlsian games, where a D\ae mon is tasked with creating a world, and an Angel places them within it. We consider various modifications and restrictions of this basic setup by adjusting the action space of the agents, as well as the payoff function, and show that various game theoretic solution concepts, including adversarial play for constant sum games and Nash equilibria for general sum games, give rise to various welfare and malfare concepts.
Of course, this game is a metaphor, but it is strongly motivated by the grounded social planner's problem, wherein a social planner seeks to organize society in a way that is favorable to all, and from these games we derive insight as to how the social planner should behave.
The goals of this paper and the purpose of constructing this game are manifold.

\begin{enumerate}
\item
We provide philosophical insight into a large class of welfare and malfare functions.
\Cref{sec:phil:mixture} draws connections between fairness and robustness, finding that many classical welfare
functions can be understood as robust utilitarian welfare in our game. We also show in \cref{sec:phil:altruistic-daemon} that some welfare (malfare) concepts arise from a class of concave utility transforms (convex disutility transforms).
These derivations complement direct fairness-based understandings of these welfare concepts from cardinal welfare theory. 
Viewing the prism of fairness from these three angles yields deeper philosophical, mathematical, and algorithmic understanding.

\item
We argue that utilitarian and egalitarian welfare or malfare are two ends of a spectrum, and derive a novel class of welfare (malfare) functions, which we term the Gini power-mean class, that falls between these extremes.
In particular, both the Gini and power-mean classes also have this property, and our Gini power-mean class \emph{strictly contains} utilitarian, egalitarian, and the entire Gini and power-mean families.
Furthermore, our generic robustness analysis allows us to define and motivate robust variants of this class that are still more general.

\item
Leveraging the connections between fairness, robustness, and robust fairness, we show in \cref{sec:adv-opt} that for various applications in allocation and machine learning, our objectives can be efficiently optimized.
In particular, we consider optimization either through standard maximin optimization techniques, or in some cases by reduction to a simpler maximization problem.
We show mild conditions under which the robust optimization problem has convex-concave structure amenable to first-order (gradient descent-ascent) methods.
We also show certain special cases in continuous allocation problems that reduce to linear programming or quadratic programming.\todo{geometric?}

\end{enumerate}

\todo{
Dual motivation: philosophical and mathematical.
\\
Cover p?
For any pmean 
Don't know weights.
\\
TODO $\displaystyle\maxmin_{\theta \in \Theta}$ $\argmaxmin$
\\
TODO We allow degenerate simplices, repair pmean as expected, define all weighted aggregator functions with the ``weightlessness axiom''
\[
\forall c \geq 0: \ \Mean(\lv; \wv) = \Mean(c \bm{1} \circ \lv; \bm{0} \circ \lv)  
\]
NB power-means are not continuous in $\wv$ around the boundary, i.e., for $\lv_{i} = 0$, $\wv_{i} \to 0$ can produce discontinuity.
}

%
It is important to note that robustness and fairness are not the same thing, but they are deeply related. In \cref{sec:phil:mixture}, we derive the Gini social welfare family as the solution to a robust utilitarian optimization problem, but it's worth noting that the Gini family on its own also arises as the unique solution set to a set of cardinal welfare axioms based on fairness that really have nothing to do with robustness.
Similarly, egalitarian or ``worst case over groups'' objectives inherently have some robustness, the form of optimizing them resembles robust objectives, and their derivation via Rawls' original position argument (\cref{sec:phil}) resembles robustness against an adversary, they can also be axiomatically derived from cardinal welfare theory, if we start with Gini axioms, then strengthen them to require that transferring utility from any group $i$ to any group $j$ with lower utility is beneficial, even if the transfer is \emph{arbitrarily inefficient}, i.e., for all $i,j$ such that $\lv_{i} < \lv_{j}$, for any \emph{transfer efficiency} $\gamma > 0$, there exists some \emph{transfer magnitude} $\alpha$ such that
\[
\Mean(\lv; \wv) < \Mean(\lv + \alpha\gamma\1_{i} - \alpha\1_{j})
\enspace,
\]
i.e., any infinitesimal equitable transfer of utility is beneficial, no matter how inefficient.
\todo{Explain this.}

\paragraph{Overview of Contributions}%
We describe in \cref{sec:phil} John Rawls' original position argument, followed by several generalizations in which the adversary is weakened, which give rise to various robust fairness concepts, including 
the utilitarian, Gini, and power-mean welfare concepts, and robust variants thereof.
In \cref{sec:math}, we show that our robust proxies of the standard fairness concepts yield probabilistic or adversarial guarantees in terms of their non-robust counterparts.
This mathematical motivation complements the philosophical motivation of the previous section.
\ifextensions
Then, \cref{sec:adv-opt} shows that we can efficiently optimize these fair and robust fair objectives in a variety of allocation and machine learning settings.
Finally, in \cref{sec:stat-fml} we analyze the continuity properties of robust fair objectives, and we show 
generalization bounds for robust fair machine learning tasks.
\fi

\bigskip

\todo{
\emph{All proofs are derived in the appendix, and all images are AI generated and thus not subject to copyright.
The authors place the blame for any intellectual property rights infringement solely at the hands of the creators of the AI.}
}

\todo{non game-theoretic: soft robust? other solution concepts? Percentile criteria, VaR?}

\section{Related Work}
\label{sec:related}

In his seminal work, \citet{rawls1971atheory,rawls2001justice} connects fairness, justice, social welfare, and robustness to uncertainty
.
His \emph{original position} or \emph{veil of ignorance} arguments apply Wald's maximin principle to derive the \emph{egalitarian welfare}, i.e., the principal that we should measure the overall wellbeing of society in terms of its least well-off member, and the social planner should seek to maximize this minimum utility.
The Rawlsian school of thought contrasts the earlier prevailing utilitarian theory: \emph{Utilitarian welfare} \citep{bentham1789introduction,mill1863utilitarianism} instead measures overall wellbeing as the \emph{sum} or \emph{average} 
utility across a population.

However, these are not the only justice criteria of interest.
Alternative characterizations of welfare lead to the power-mean class or the Gini class (discussed in \cref{sec:prelim}, both of which contain the egalitarian and utilitarian welfare as special cases.
Indeed, the wellbeing of society overall and of disadvantaged or minority groups
is well-studied
in welfare economics \citep{pigou1912wealth,dalton1920measurement,debreu1959topological,gorman1968structure} and moral philosophy \citep{parfit1997equality}. 
Generally speaking, utilitarian and egalitarian welfare stand at two extremes of a spectrum, and \emph{prioritarian} concepts lie somewhere in between \citep{parfit1997equality,arneson2000luck}.
Utilitarianism is criticized for not incentivizing \emph{equitable redistribution} of (dis)utility, and egalitarianism is criticized for 
ignoring
all but the 
\emph{most disadvantaged} groups
in society.
In contrast,
\emph{prioritarianism}
encompasses
various justice criteria that \emph{prioritize} the wellbeing of
the impoverished, without ignoring 
others,
making tradeoffs between 
them
in various ways.

\Citet{amadae2003rationalizing,galivsanka2017just} discuss the game-theoretic implication of Rawlsian philosophy, and this work considers modifications of a game-theoretic statement of Rawls' original position argument.
\Citet{nozick1974anarchy} criticizes Rawlsian theory as overly risk-averse, and
this work addresses 
this point by introducing less risk-averse variants of the original position argument.
Rawlsian theory is also criticized as unsuitable as a basis for morality \citep{harsanyi1975can}; this work makes progress in this direction, by explicitly grounding the application of Wald's \citeyearpar{wald1939contributions,wald1945statistical} maximin principle in epistemic uncertainty, and by contrasting Rawlsian welfare with other fairness concepts.

Rawlsian theory has been applied to fair machine learning and algorithmic justice \citep{ashrafian2023engineering}, often termed \emph{minimax fair learning} \citep{diana2021minimax,shekhar2021adaptive,abernethy2022active}, and in particular \citet{lokhande2022towards,dong2022decentering} 
optimize Rawlsian objectives under uncertainty.
More general concepts of fair machine learning in terms of other welfare or malfare concepts are also explored in the literature; \citet{thomas2019preventing} introduces the \emph{Seldonian learner} framework, 
and \citet{cousins2021axiomatic,cousins2021bounds,cousins2022uncertainty,cousins2023revisiting} defines \emph{fair-PAC learning}, which both deal with computational and statistical issues arising from optimizing nonlinear fairness objectives.

Other work seeks to optimize welfare-concepts in more specific instances; e.g., in supervised classification \citep{hu2020fair,rolf2020balancing}, in contextual bandits \citep{metevier2019offline}, or in reinforcement learning \citep{siddique2020learning,cousins2022faire3}, generally finding the resulting optimization problems to be tractable.
Some authors also seek to apply \emph{fairness constraints} based on welfare \citep{hu2020fair,heidari2018fairness,speicher2018unified}; \citet{hu2020fair} finds that, in contrast to demographic parity constraints, these are usually at least \emph{convex} (assuming appropriate utility and welfare function choice).

\todo{cite \citep{} {kasy2021fairness} }

\todo{Discuss bilevel.}

\todo{More seldonian citations.}

\section{Preliminaries}
\label{sec:prelim}

\todo{\paragraph{
On wefare-centric fair ML}}

The \emph{Pigou-Dalton transfer principle} \citep{pigou1912wealth,dalton1920measurement} and
the \emph{Debreu-Gorman axioms} \citep{debreu1959topological,gorman1968structure} 
lead all welfare functions to 
concord with
sums of \emph{logarithms} or \emph{powers} of utilities, i.e.,
for $\NGroups$ groups and
utility vectors $\lv \in \RNN^{\NGroups}$,
for some $p \in \R$,
all fairness concepts $\Mean(\lv)$ 
define a \emph{partial ordering} over utility vectors that agrees with
\begin{equation}
\label{eq:dg}
\Mean(\lv) = \sgn(p) \sum_{i=1}^{\NGroups} \lv_{i}^{p}
\enspace, \quad \text{or} \quad
\Mean(\lv) = \sum_{i=1}^{\NGroups} \ln( \lv_{i} )
\enspace.
\end{equation}

Weights vectors $\wv \in \triangle_{\NGroups}$, where $\triangle_{\NGroups}$
denotes
the unit probability simplex over $\NGroups$ values (excluding $0$ values\todo{check}), are essential to this work. 
\citet{cousins2021axiomatic,cousins2023revisiting} introduces
weighted variants of the Debreu-Gorman axioms, as well as \emph{multiplicative linearity} and \emph{unit scale} axioms, which essentially standardize the cardinal values of aggregator functions.
Utility and disutility are generically referred to as \emph{sentiment}, and cardinal welfare theory applies equally well to aggregation of disutility (malfare functions).
These novel axioms, when combined with the Debreu-Gorman axioms, characterize the 
\emph{weighted power-mean family} of aggregator functions, defined below. 
%
\begin{definition}[Weighted Power-Mean Family]
Suppose some power parameter $p \in \R$ and weights vector $\wv \in \triangle_{\NGroups}$. For any $\lv \in \RNN^{\NGroups}$, we define 
\begin{equation}
\label{eq:pmean}
\Mean_{p}(\lv; \wv) = \sqrt[p]{ \vphantom{\sum^{.}_{}} \smash{ \sum_{i=1}^{\substack{\NGroups\\[-0.55ex]}} } \wv_{i} \lv_{i}^{p} } \ \text{for $p \neq 0$}
\enspace, \quad 
\Mean_{0}(\lv; \wv)
  = \exp \left( \sum_{i=1}^{\NGroups} \wv_{i} \ln( \lv_{i} ) \right)
\enspace, \quad \text{or} \quad
\Mean_{\pm\infty}(\lv; \wv) = \maxmin_{\mathclap{i \in 1, \dots, \NGroups}} \lv_{i}
\enspace.
\end{equation}
The taking the limit as $p \to 0$ yields the $p = 0$ case, 
known as the \emph{Nash social welfare} or \emph{geometric mean}, and the limits as $p \to \pm \infty$ yield the \emph{egalitarian} welfare or malfare.
\end{definition}

Power-means with $p \in (-\infty, 1)$ are valid welfare functions, as maximizing them strictly incentivizes equitable redistribution of utility (except around $0$).
Similarly, power-means with $p \in (1, \infty)$ are valid malfare functions, as minimizing them strictly incentivizes equitable redistribution of disutility.
The utilitarian and egalitarian endpoints of these open intervals ($p \in \{-\infty, 1, \infty\}$) are generally also considered valid, as they arise as limiting sequences of valid welfare or malfare functions, they still satisfy most of the same cardinal welfare axioms as the power-mean family, and they at least weakly incentivize equitable redistribution.

\ifextensions
Power-means in general require \emph{nonnegative sentiment} to remain well-defined, real-valued, and preserve their curvature, so when working with them we restrict the sentiment vector space to $\lv \in \RNN^{\NGroups}$, but for other aggregator function classes, we may relax this assumption to $\lv \in \R^{\NGroups}$.
\todo{Alternatively, we can extend the power-mean to the entire real domain $\R^{\NGroups}$ 
as
\[
\Mean_{p}(\lv; \wv) \doteq
\begin{cases}
p \in \pm\infty & \maxmin_{i \in 1, \dots, \NGroups} \lv_{i} \\
p = 1 & \lv \cdot \wv \\
p \in (-\infty, 1) & -\infty \\
p \in (1, \infty) & \Mean_{p}(\bm{0} \vee, \lv; \wv) 
\end{cases}
\]
\todo{Handle $p < 0$ better?}
for any $\lv$ 
containing negative sentiment values.
These are the most crucial properties for our purposes, however there are some undesirable artifacts introduced here; for example, the codomain now includes $-\infty$ to preserve concavity of welfare functions, discontinuities in $p$ are introduced at $-\infty, 1, \infty\}$, and due to these discontinuities, the function no longer matches its limits as $p$ approaches these values.
\todo{Differentiability in $p$, however, is lost.}
\\
TODO BETTER EXTENSION:
\[
\Mean_{p}(\lv; \wv) \doteq \Mean_{p}(\bm{0} \vee \lv; \wv) + (\bm{0} \wedge \lv) \cdot \grad^{+} \Mean(\bm{0} \wedge \lv; \wv)
\]
where $\grad^{+} \Mean(\bm{0} \wedge \lv; \wv)$ refers to the vector of \emph{right derivatives} evaluated at the point $\bm{0} \wedge \lv$.
NB: 
only problem is for all negative or all $\bm{0}$.
\[
\begin{cases}
p \in \pm\infty & \maxmin_{i \in 1, \dots, \NGroups} \lv_{i} \\
p = 1 & \lv \cdot \wv \\
p \in (-\infty, 1) & -\infty \\
p \in (1, \infty) & \Mean_{p}(\bm{0} \vee, \lv; \wv) 
\end{cases}
\]
}
\fi

\todo{Sentiment word!}

A slightly different set of axioms 
yields the Gini class \citep{weymark1981generalized,gajdos2005multidimensional}.
\begin{definition}[Gini Welfare and Malfare]
Suppose a \emph{ascending sequence} $\wv^{\uparrow} \in \triangle_{\NGroups}$ or \emph{descending sequence} of \emph{Gini weights} $\wv^{\downarrow} \in \triangle_{\NGroups}$, \emph{
risk vector} $\lv \in \R^{\NGroups}$, and let $\lv^{\downarrow}$ denote $\lv$ in descending 
order.
The generalized Gini social welfare function (GGSWF) 
is then
\begin{equation}
\label{eq:gini-welfare}
\Mean_{\wv^{\uparrow}}(\lv) \doteq \vphantom{\sum}\smash{\sum_{i=1}^{\NGroups}} \wv^{\uparrow}_{i} \lv^{\downarrow}_{i} = \wv^{\uparrow} \cdot \lv^{\downarrow}
\enspace,
\end{equation}
and similarly, the Gini social malfare function is
\begin{equation}
\label{eq:gini-malfare}
\Mean_{\wv^{\downarrow}}(\lv) \doteq \vphantom{\sum}\smash{\sum_{i=1}^{\NGroups}} \wv^{\downarrow}_{i} \lv^{\downarrow}_{i} = \wv^{\downarrow} \cdot \lv^{\downarrow}
\enspace.
\end{equation}
\todo{Combine lines.}
\end{definition}

Notably, while the power-mean family is not closed under convex combination, the Gini family is. Furthermore, restricting to convex combinations of \emph{utilitarian} and \emph{egalitarian} welfare or malfare yields the \emph{utilitarian-maximin} social welfare function (UMSWF) family.
Several axiomatizations for the UMSWF class exist in the literature \citep{deschamps1978leximin,bossert2020axiomatization,schneider2020utilitarian}, each essentially strengthening Gini axioms in some way.

\todo{
Both
utilitarian and egalitarian malfare arise as 
power-mean special-cases, 
namely $p=1$ and 
$p=\infty$, respectively.
They also arise in the Gini family, for \scalebox{0.9}[0.95]{$\wv^{\uparrow} \! = \! \langle \mathsmaller{\frac{1}{\NGroups}}, \dots, \mathsmaller{\frac{1}{\NGroups}} \rangle$} and \scalebox{0.9}[0.95]{$\wv^{\uparrow} \! = \! \langle 1, 0, \dots, 0 \rangle$}, respectively.
}
\todo{swf term.}

\section{A Philosophy of Robust Fair Objectives}
\label{sec:phil}

\begin{figure}
\centering
\begin{minipage}{0.495\textwidth}
\noindent\includegraphics[width=\textwidth]{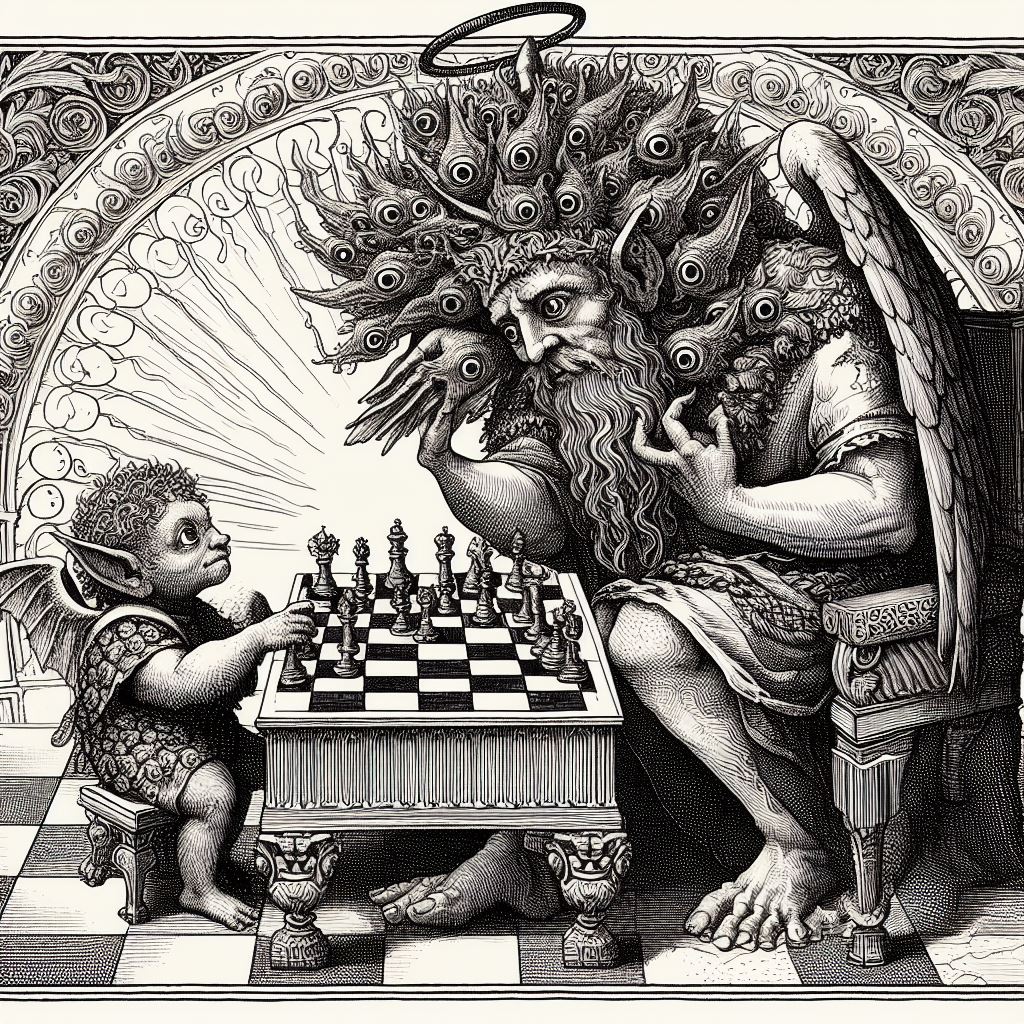}
\end{minipage}%
\hfill%
\begin{minipage}{0.495\textwidth}

\ifforc
\small
\renewcommand{\Large}{\large}
\renewcommand{\medskip}{\smallskip}
\fi

\centering 
{ \Large\titlefont The Rawlsian Game }

\medskip

$\NGroups$: \emph{Number of inhabitants} in the game world, labeled $1, \dots, \NGroups$.

$\Theta$: Set of \emph{feasible parameters} for worlds the D\ae mon can create.


$\lv(\theta): \Theta \to \R^{\NGroups}$:
Sentiment vector of the
inhabitants of
some world parameterized by $\theta$.

$\Lv \doteq \left\{ \lv(\theta) \middle| \, \theta \in \Theta \right\} \subseteq \R^{\NGroups}$: Set of \emph{feasible
sentiment
vectors}
the D\ae mon can create. 

$\mathcal{A}_{\Dae} \doteq \Lv$: D\ae mon action space. 

$\mathcal{A}_{\Ang} \doteq \{1, \dots, \NGroups\}$: Angel action space. 

$\Payoff(\lv, i) \doteq \langle \lv_{i}, -\lv_{i} \rangle$:
Zero-sum
payoff function.

\medskip

Strategic gameplay (D\ae mon goes first):
\vspace{-0.25cm}
\[
\argmaxmin_{\lv \in \mathcal{A}_{\Dae}} \minmax_{i \in \mathcal{A}_{\Ang}} \Payoff_{1}(\lv, i)
= \argmaxmin_{\lv \in \Lv} \minmax_{i \in 1, \dots, \NGroups} \lv_{i} \enspace.
\]
\end{minipage}

\caption{
Metaphoric depiction and game-theoretic description of the Rawlsian original position game.
A
weak
D\ae mon (left) plays against an all-seeing Angel (right).
}
\label{fig:rawlsian-game}
\end{figure}

Inspired by 
the original position argument and the veil of ignorance of 
\citet{rawls1971atheory,rawls2001justice},\footnote{Note that Rawls' original position argument is generally phrased in terms of Wald's maximin principle and robustness to uncertainty, rather than explicitly as a
zero-sum game. 
These characterizations are equivalent, and for our purposes, it is often convenient to characterize uncertainty as the action space of an explicit adversary. In general, this is to simplify intuition and the use of standard tools from game theory; it is not meant to be interpreted as a literal adversary in a literal game.}
we pose a series of adversarial games, where a D\ae mon is tasked with creating a world, and an Angel then punishes the D\ae mon by choosing whom to reincarnate them as in their world.
In many ways, this is an unfair game, as the D\ae mon is given an impossibly difficult task, and the Angel lazily stands by until it is their turn to inflict maximal suffering upon their opponent%
\iffalse\else, but perhaps it is an allegory for the responsibility of political leaders, and were they to face harsher rebuke from the citizenry, perhaps we would live in a more equitable society\fi.
We stress that this metaphor does not represent a conflict between good and evil, but rather a cosmic struggle between the
freedom of the people (as represented by the D\ae mon) and dictatorial power (as represented by the Angel).

Because we wish to treat both utility and disutility (with welfare and malfare), we generically refer to these concepts as \emph{sentiment}, and we adopt neutral notation $\lv$ to represent
sentiment
vectors. 
In general, we use stacked operators, e.g., $\pm, \mp, \maxmin, \infsup$, etc.\ to represent both cases simultaneously.
Unless otherwise noted, the upper operator describes the utility branch and the lower operator is for the disutility branch.

Furthermore, to distinguish between the the utility values of the inhabitants of the game world $\lv$, and those of the D\ae mon and Angel playing the game, we refer to the latter as a payoff function $\Payoff(\lv; \wv)$, representing either positive payoff in the utility case or a negative payoff in the disutility case.
In this adversarial zero-sum game, the D\ae mon's payoff is the utility of the person they become, and the Angel's payoff is of course its negation.

The Angel's adversarial response to any D\ae mon strategy is obvious:
\emph{Select the individual with the lowest utility (or highest disutility)}.
Playing against this Angel, the D\ae mon must confront the question, \emph{``How should we construct a world without knowing our place in it?''}
Against an adversarial Angel, a strategic D\ae mon must maximize the minimum utility (or minimize the maximum disutility).
From this interaction, strategic gameplay results in the solution concept
\[
\argmaxmin_{\lv \in \mathcal{A}_{\Dae}} \minmax_{i \in \mathcal{A}_{\Ang}} \Payoff_{1}(\lv, i) =
\argmaxmin_{\lv \in \Lv} \minmax_{i \in 1, \dots, \NGroups} \lv_{i} = \begin{cases}
\text{Utility} & \displaystyle \max_{\lv \in \Lv} \min_{1 \in 1, \dots, \NGroups} \lv_{i} = \argmax_{\lv \in \Lv} \Mean_{-\infty}(\lv) \\
\text{Disutility} & \displaystyle \min_{\lv \in \Lv} \max_{1 \in 1, \dots, \NGroups} \lv_{i} = \argmin_{\lv \in \Lv} \Mean_{\infty}(\lv) \enspace. 
\end{cases}
\]
The game is illustrated and further described in \cref{fig:rawlsian-game}.

The power of the D\ae mon in this game is directly modeled by the scope of worlds that they are capable of creating, and the Angel's power is directly impacted by the number of people that inhabit
the world.
We then consider variations of this game where the Angel is weakened in various ways, and show that they give rise to other standard notions of welfare, in the sense that the D\ae mon's optimal strategy is to maximize some welfare concept.
We also observe that in many cases, the original game is equivalent to one where the Angel must move first, but is allowed to employ a randomized strategy.
We show that our modifications to the game can be formulated in several equivalent ways, and should not be taken too literally, as one motivation or another may be more or less suitable depending on the context of applications or philosophical stance of the reader.

\subsection{On Mixed Strategies and Weak (Constrained) Adversaries}
\label{sec:phil:mixture}

\begin{figure}
    \centering

    \begin{minipage}{0.495\textwidth}
\includegraphics[width=-\textwidth,height=\textwidth]{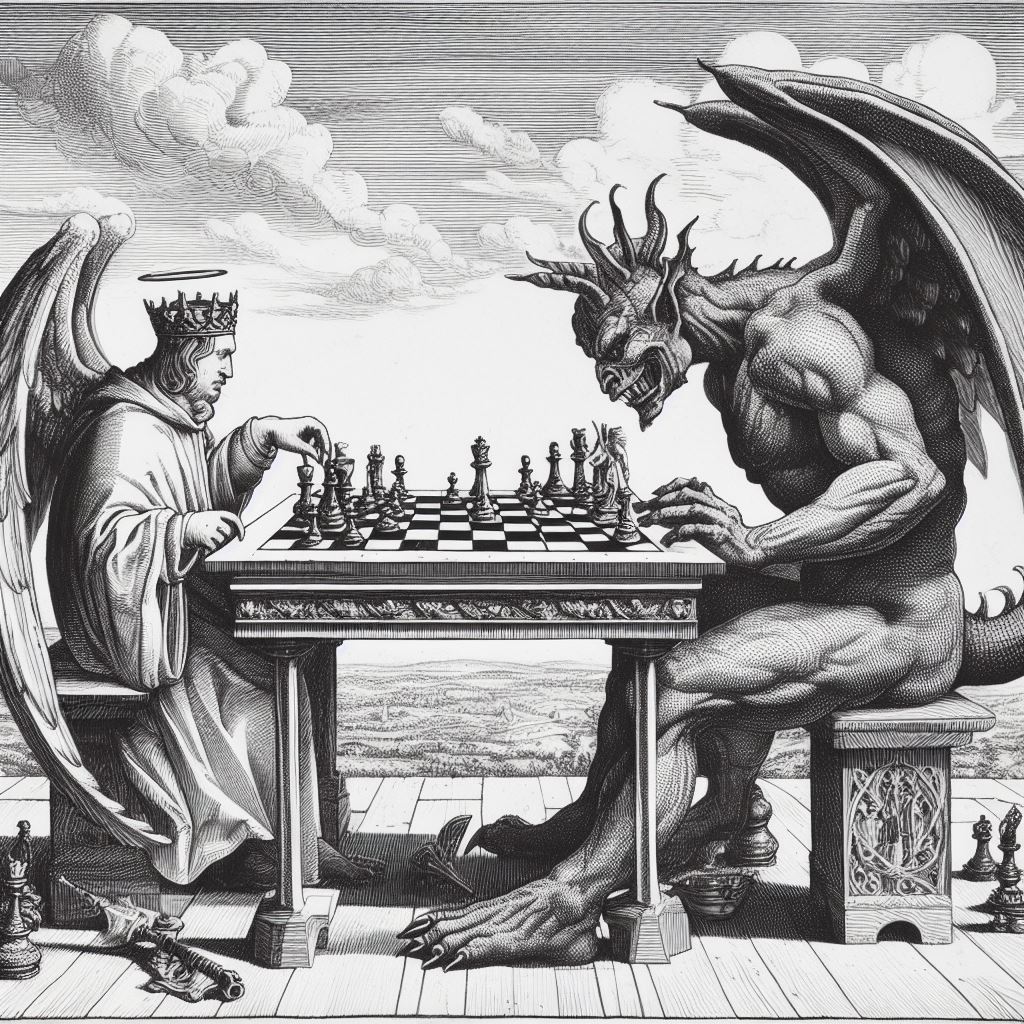}
    \end{minipage}%
    \hfill
    \begin{minipage}{0.495\textwidth}

\ifforc
\small
\renewcommand{\Large}{\large}
\renewcommand{\medskip}{\smallskip}
\fi

\centering 
{ \Large\titlefont The Constrained Rawlsian Game }

\medskip

$\NGroups, \Theta, \lv(\theta), \Lv$: Group count, D\ae mon parameter space, sentiment function,
sentiment
space (as in \cref{fig:rawlsian-game}). 

$\Wv \subseteq \triangle_{\NGroups}$: Constrained weights space.

$\mathcal{A}_{\Dae} \doteq \Lv$: D\ae mon action space. 

$\mathcal{A}_{\Ang} \doteq \Wv$: Angel action space \\ 
(\emph{inhabitant} $i$ becomes \emph{distribution
} $\wv$).

$\Payoff(\lv; \wv): \R^{\NGroups} \times \triangle_{\NGroups} \to \R^{2}$: 
Zero-sum payoff function (expected sentiment).
\[
\Payoff_{1} (\lv, \wv) \doteq \langle \wv \cdot \lv, - \wv \cdot \lv \rangle 
\]

\medskip

Strategic gameplay (D\ae mon goes first):
\vspace{-0.25cm}
\[
\argmaxmin_{\lv \in \mathcal{A}_{\Dae}} \minmax_{\wv \in \mathcal{A}_{\Ang}} \Payoff_{1}(\lv, i)
= \argmaxmin_{\lv \in \Lv} \minmax_{\wv \in \Wv} \wv \cdot \lv \enspace.
\]
    \end{minipage}

    \caption{%
Metaphoric depiction and game-theoretic description of the modified Rawlsian original position game, with restricted Angel action space.\todo{terminology!}
A D\ae mon (left) plays against a comparably powerful Angel (right).
    }
    \label{fig:constrained-rawlsian-game}
\end{figure}

The power of the D\ae mon in this game is directly controlled by the space $\Lv$ of feasible utility vectors (or some generating parameter space $\Theta$), and allowing the D\ae mon to play mixed strategies simply expands their action space to the convex hull $\CH(\Lv)$ (assuming the payoff function is defined as their \emph{expected} utility). 
At times, we may require $\Lv$ to be a convex set, for which mixed D\ae mon strategies are sufficient but not necessary.
We explicitly assume as such when necessary, thus we make no further mention of mixed strategies on the D\ae mon's part.

In our game, the Angel is allowed to condition their action on the D\ae mon's \emph{mixture action}, but not on the actual D\ae mon action randomly selected from this mixture, thus when $\Lv \neq \CH(\Lv)$, playing mixed strategies may increase the D\ae mon's power in this game.
However, because the Angel plays second, mixed strategies do not actually increase the Angel's power.
In particular, the expected value of the payoff function given a mixed D\ae mon strategy (distribution over (dis)utility vectors) $\ProbDist \in \mathcal{A}_{\Dae}$ and an (independent) mixed Angel strategy (distribution over group indices, i.e., $\wv \in \triangle_{\NGroups}$) is
\begin{equation}
\label{eq:expected-payoff}
\Payoff_{1} (\ProbDist, \wv) = \Expect_{\lv \distributed \ProbDist \bot i \distributed \wv} \left[ \Payoff_{1} (\lv, i) \right] = \Expect_{\lv \distributed \ProbDist \bot i \distributed \wv} \left[ \lv_{i} \right] = \Expect_{\lv \distributed \ProbDist \bot i \distributed \wv} \left[ \1_{i} \cdot \lv \right] = \wv \cdot \Expect_{\lv \distributed \ProbDist} \bigl[ \lv \bigr] \enspace,
\end{equation}
where $A \bot B$ denotes \emph{independence of random variables}.
Strategic 
gameplay then
results in the objective value (expected D\ae mon payoff)
\begin{equation}
\label{eq:mixed-strategic-play}
\supinf_{\ProbDist \in \mathcal{A}_{\Dae}} \maxmin_{\wv \in \mathcal{A}_{\Ang}} \Payoff_{1}(\ProbDist, \wv) = \supinf_{\ProbDist \in \mathcal{A}_{\Dae}} \maxmin_{\wv \in \mathcal{A}_{\Ang}} \wv \cdot \Expect_{\lv \distributed \ProbDist} \bigl[ \lv \bigr] = \supinf_{\lv \in \CH(\lv)} \maxmin_{\wv \in \Wv} \wv \cdot \lv \enspace.
\end{equation}

\paragraph{Action Order Interchangeability}
It is worth noting that, while the order of play matters for specific strategies, under mild conditions, assuming strategic play, the order is interchangeable (and thus the game may be played with simultaneous actions).
In particular, if we allow for mixed actions (or at least a convex action set for the D\ae mon), we have the following result.
\begin{lemma}[
Maximin
Interchangeability]
Suppose that both $\mathcal{A}_{\Dae}$ and $\mathcal{A}_{\Ang}$ are convex sets and $\mathcal{A}_{\Ang}$ is closed (thus also compact).
Then 
\begin{equation}
\label{eq:maximin-interchangeability}
\supinf_{\ProbDist \in \mathcal{A}_{\Dae}} \maxmin_{\wv \in \mathcal{A}_{\Ang}} \Payoff_{1}(\ProbDist, \wv) = \maxmin_{\wv \in \mathcal{A}_{\Ang}} \supinf_{\ProbDist \in \mathcal{A}_{\Dae}} \Payoff_{1} (\ProbDist, \wv) \enspace.
\end{equation}

\end{lemma}
\begin{proof}
This result follows from \eqref{eq:mixed-strategic-play}, then application of Sion's \citeyearpar{sion1958general} minimax theorem.
\end{proof}

\paragraph{On Special Cases}
We now show that many of the most commonly employed welfare and malfare functions 
arise as special cases of this constrained-Angel Rawlsian game.
This intuitively motivates these aggregator functions from the perspective of robustness, with no robustness representing utilitarianism, establishing a spectrum of various degrees and types of robustness, with egalitarianism at the opposite end of the spectrum (as in the power-mean and Gini classes).
Depending on how natural the choice of constrained weight space is, it may lend credence to the use of particular aggregator functions, and this analysis is also relevant to
their optimization or analysis,
as they can be treated with standard maximinimization\todo{maximin or minimax? Be consistent!} tools, discussed further in \cref{sec:adv-opt}.
We also note that this robustness interpretation merely \emph{complements} (but does not replace or invalidate) existing fairness interpretations of these aggregator functions, and in \cref{sec:phil:altruistic-daemon,sec:phil:altruistic-angel}, we show other categories of fair and fair robust objective that don't seem to arise from just the robustness aspect of this particular game.

For the purposes of this characterization, we consider any fixed D\ae mon strategy (or mixture of strategies), as represented by some (expected) (dis)utility vector $\lv$.
The following result characterizes the adversarial (worst case) payoff of the D\ae mon against the Angel, i.e., we show that $\minmax_{\wv \in \Wv} \Payoff_{1} (\lv; \wv) = \Mean(\lv)$ for some classical aggregator function $\Mean(\lv)$, thus the D\ae mon's optimal strategy in this game is to optimize these aggregator functions.
This generalizes the maximum principle of Rawlsian theory, wherein from behind the veil of ignorance, the social planner chooses to \emph{maximize} the \emph{minimum} utility, i.e., egalitarian welfare.

\begin{theorem}[Classical Welfare and Malfare Functions as Constrained Angel Solution Concepts]
\label{thm:special-cases-classical}

For any $\lv \in \R^{\NGroups}$, the 
Angel's best responses under the following special cases
of Angel action spaces $\mathcal{A}_{\Ang}$, as represented by
weight action spaces $\Wv$, give rise to 
standard aggregator functions.
Some cases assume fixed ``true weights'' $\wv^{*} \in \triangle_{\NGroups}$ or
Gini weights sequence $\wv^{\downarrow} \in \triangle_{\NGroups}$ or $\wv^{\uparrow} \in \triangle_{\NGroups}$, and obtain special cases in terms of these parameters.

\begin{enumerate}
\item\label{thm:special-cases-classical:egal} \textbf{Egalitarian}:
Suppose $\Wv = \triangle_{\NGroups}$.
Then
$
\displaystyle
\minmax_{\wv \in \Wv} \Payoff_{1} (\lv; \wv) =
 \Mean_{\mp \infty}(\lv)
\enspace.
$


\item\label{thm:special-cases-classical:util} \textbf{Utilitarian}: Suppose $\Wv = \{\wv^{*}\}$.
Then
$
\displaystyle
\minmax_{\wv \in \Wv} \Payoff_{1} (\lv; \wv) = \Mean_{1}(\lv; \wv^{*})
\enspace.
$

\item\label{thm:special-cases-classical:wumswf} \textbf{Weighted Utilitarian-Maximin
}: Suppose $\Wv = \left\{ \wv \, \middle| \, \wv \succeq \gamma \wv^{*} \right\} = \gamma \{ \wv^{*} \} + (1 - \gamma) \triangle_{\NGroups}$.
\\
Then
\linebreak[3]$
\displaystyle
\minmax_{\wv \in \Wv} \Payoff_{1} (\lv; \wv) = \gamma \Mean_{1}(\lv; \wv^{*}) + (1 - \gamma) \Mean_{\mp\infty}(\lv)
\enspace.
$

\item\label{thm:special-cases-classical:ggswf} \textbf{Generalized Gini%
}: Suppose
$\Wv = \left\{ \pi(\wv^{\downarrow}) \middle| \pi \in \Pi_{\NGroups} \right\}$, 
where $\Pi_{\NGroups}$ is the set of all permutations on $\NGroups$ items.
Then
\linebreak[3]$
\displaystyle
\min_{\wv \in \Wv} \Payoff_{1} (\lv; \wv) = \Mean_{\wv^{\uparrow}}(\lv)
$ or $
\displaystyle
\max_{\wv \in \Wv} \Payoff_{1} (\lv; \wv) = \Mean_{\wv^{\downarrow}}(\lv)
\enspace.
$

\end{enumerate}
Furthermore, for each of the above items, the RHS follows for any Angel action space $\Wv'$,
$\Wv$ and $\Wv'$ have the same convex hull,
e.g., for \cref{thm:special-cases-classical:egal}, we may use $\left\{ \1_{i} \, \middle| \, i \in 1, \dots, \NGroups \right\}$ in place of $\triangle_{\NGroups}$.
With this expansion, a sort of converse follows for each result.
For each of the above items, if 
the conclusion holds
\emph{for all $\lv \in \R^{\NGroups}$}, then $\Wv$ is some such $\Wv'$.
For example, from item~1 we have
\linebreak[3]
$\displaystyle
\minmax_{\wv \in \Wv} \Payoff_{1}(\lv; \wv) =
 \Mean_{\mp \infty}(\lv) \implies \left\{ \1_{i} \, \middle| \, i \in 1, \dots, \NGroups \right\} \subseteq \Wv \subseteq \triangle_{\NGroups}
\enspace.
$
%
%
\end{theorem}

\Cref{thm:special-cases-classical:egal} is rather obvious, as this special case is just the unconstrained game.
\Cref{thm:special-cases-classical:util} is also unsurprising, as this special case replaces the robust or worst case perspective of the Rawlsian game with an \emph{average case} or expected perspective,
which
yields \emph{weighted sums} of (dis)utility,
i.e.,
utilitarian malfare or welfare.
Of course, uniform
individual weights $\wv = \langle \frac{1}{\NGroups}, \frac{1}{\NGroups}, \dots, \frac{1}{\NGroups} \rangle$ correspond to \emph{uniformly randomly} selecting among all living individuals, which very much concords with the utilitarian perspective.\footnote{Furthermore, assuming nonuniform \emph{group weights} $\wv$ correspond to the \emph{population frequencies} of each group, this perspective still replaces the risk-aversion of Wald's maximin principle with a \emph{uniform average} over all individuals.} 
Despite the mathematical simplicity of these results, it is encouraging to see that the two most popular aggregator functions do arise as special cases of this adversarial game, and in some sense they are the extreme cases, as the Angel action space is \emph{maximal} (complete) in \cref{thm:special-cases-classical:egal} and \emph{minimal}
(singleton)
in \cref{thm:special-cases-classical:util}.

In contrast, \cref{thm:special-cases-classical:wumswf} is rather surprising, as we see that a simple lower-bound constraint on weights values produce the classical (weighted) utilitarian maximin social welfare function (UMSWF).
The statement of the result gives some intuition: this Angel action space is a convex combination of the egalitarian and utilitarian action spaces, and so too is the weighted UMSWF a convex combination of the egalitarian and utilitarian aggregator functions.
Also of note is that
the unweighted UMSWF is a more restrictive class then the GGSWF, and is theoretically justified by a rather heavy-handed (strong) set of axioms.
This result gives an alternative characterization of UMSWF as a robust variant of utilitarian welfare, where $\gamma \wv_{i}^{*}$ is a lower bound on the 
weight of each group $i$ (note that
WLOG any such set of feasible lower-bounds can be represented for some $\gamma \in [0, 1]$, $\wv^{*} \in \triangle_{\NGroups}$).
Finally, \cref{thm:special-cases-classical:ggswf} is perhaps the most sophisticated result here, as the class of Angel actions has quite a bit more structure (though it is still a bounded polytope).
This characterization also provides an alternative characterization of the Gini social welfare as a robust utilitarian objective, where the weights relative population sizes of all groups are known, but the identities of the group associated with each weight is not known.

\todo{leximin?}

We now show that expanding the classes of \cref{thm:special-cases-classical} (
except for egalitarian, which is already maximal) results in novel nontrivial robust aggregator function concepts.
Each of these robust aggregators essentially optimizes a classical aggregator function subject to a worst case assumption \wrt\ some type of uncertainty.
We illustrate in \cref{fig:simplex-constraints} robustness sets defined by the $\mathcal{L}_{\infty}$, $\mathcal{L}_{2}$, and $\mathcal{L}_{1}$ norms around a point $\wv^{*}$.

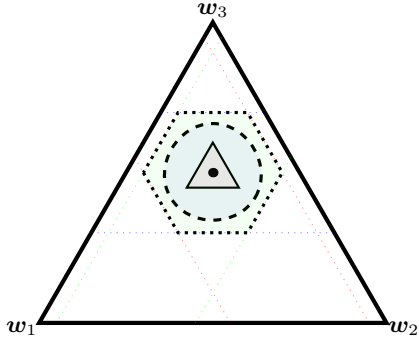
\begin{figure}

\begin{minipage}{0.349\textwidth}

\centering


    \ifforc
    \vspace{-0.45cm}
    \newcommand{\thisscale}{2.47}
    \else
    \newcommand{\thisscale}{2.666}
    \fi

\null\hspace{-0.3cm}%
\begin{tikzpicture}[
    scale=\thisscale,
    x={(-0.866025404cm,-0.5cm)},y={(0.866025404cm,-0.5cm)},z={(0cm,1cm)},
    cline/.style={very thick},
]

\draw[ultra thick] (1, 0, 0) -- (0, 1, 0) -- (0, 0, 1) -- cycle;


\ifforc
\newcommand{\thisrad}{0.121}
\else
\newcommand{\thisrad}{0.1}
\fi

\node (x) at (1+\thisrad, 0, 0+0.25*\thisrad) {\small$\wv_{1}$};
\node (y) at (0, 1+\thisrad, 0+0.25*\thisrad) {\small$\wv_{2}$};
\node (z) at (0, 0, 1+0.6*\thisrad) {\small$\wv_{3}$};

\renewcommand{\thisrad}{0.2}


\draw[line width=0.125cm,line cap=round] (0.25, 0.25, 0.5) -- (0.25, 0.25, 0.5); 


\begin{scope}[blend group=multiply]

\draw[thick,fill=red,fill opacity=0.05] (0.25+0.5*\thisrad, 0.25-0.25*\thisrad, 0.5-0.25*\thisrad) -- (0.25-0.25*\thisrad, 0.25+0.5*\thisrad, 0.5-0.25*\thisrad) -- (0.25-0.25*\thisrad, 0.25-0.25*\thisrad, 0.5+0.5*\thisrad) -- cycle;



\newcommand{\sphtr}{0.816496581} 
\draw[very thick,dashed,fill=blue,fill opacity=0.05] plot[,smooth cycle,tension=1.5] coordinates {(0.25+\sphtr*\thisrad, 0.25-0.5*\sphtr*\thisrad, 0.5-0.5*\sphtr*\thisrad) (0.25-0.5*\sphtr*\thisrad, 0.25+\sphtr*\thisrad, 0.5-0.5*\sphtr*\thisrad) (0.25-0.5*\sphtr*\thisrad, 0.25-0.5*\sphtr*\thisrad, 0.5+\sphtr*\thisrad)};

\end{scope}


\iftrue

\begin{scope}[blend group=multiply]

\draw[red,dotted,,opacity=0.5] (0.25+\thisrad, 0, 0.75-\thisrad) -- (0.25+\thisrad, 0.75-\thisrad, 0);
\draw[red,dotted,,opacity=0.5] (0.25-\thisrad, 0, 0.75+\thisrad) -- (0.25-\thisrad, 0.75+\thisrad, 0);

\draw[green,dotted,,opacity=0.5] (0, 0.25+\thisrad, 0.75-\thisrad) -- (0.75-\thisrad, 0.25+\thisrad, 0);
\draw[green,dotted,,opacity=0.5] (0, 0.25-\thisrad, 0.75+\thisrad) -- (0.75+\thisrad, 0.25-\thisrad, 0);

\draw[blue,dotted,,opacity=0.5] (0.5-\thisrad, 0, 0.5+\thisrad) -- (0, 0.5-\thisrad, 0.5+\thisrad);
\draw[blue,dotted,,opacity=0.5] (0.5+\thisrad, 0, 0.5-\thisrad) -- (0, 0.5+\thisrad, 0.5-\thisrad);

\end{scope}

\fi


\begin{scope}[blend group=multiply]

\draw[very thick,dotted,fill=green,fill opacity=0.05] (0.25+\thisrad, 0.25-\thisrad, 0.5) -- (0.25+\thisrad, 0.25, 0.5-\thisrad) -- (0.25, 0.25+\thisrad, 0.5-\thisrad) -- (0.25-\thisrad, 0.25+\thisrad, 0.5) -- (0.25-\thisrad, 0.25, 0.5+\thisrad) -- (0.25, 0.25-\thisrad, 0.5+\thisrad) -- cycle;

\end{scope}

\end{tikzpicture}

\end{minipage}%
\hfill
\begin{minipage}{0.619\textwidth}

\caption{%
A simplicial plot over $\triangle_{3}$ of the robustness sets defined by intersection with the $\mathcal{L}_{\infty}$, $\mathcal{L}_{2}$, and $\mathcal{L}_{1}$ norm balls of radius
$\frac{1}{5}$
around the point $\wv^{*} = \left\langle \frac{1}{4}, \frac{1}{4}, \frac{1}{2} \right\rangle$.
The boundaries of the $\mathcal{L}_{\infty}$, $\mathcal{L}_{2}$, and $\mathcal{L}_{1}$ balls are plotted in solid, dashed, and dotted lines, respectively.
Assuming positive radius $r$ such that 
each $\wv^{*}$-centered norm ball is
contained
by the unit hypercube,
i.e., $\norm{\wv^{*}}_{\infty} \leq \norm{\wv^{*}}_{2} \leq \norm{\wv^{*}}_{1} \leq r$,
intersection with the unit simplex yields an equilateral triangular, circular, or hexagonal region, respectively, with $\NGroups=3$.
In higher dimensions, the regions become simplicial, hyperspherical, or regular-polytopal\todo{n choose n/2 for even n?}, respectively.
}
\label{fig:simplex-constraints}
\end{minipage}

\end{figure}

\begin{theorem}[Robust Welfare and Malfare Functions as Constrained Angel Solution Concepts]
\label{thm:special-cases-robust}
Suppose as in
\cref{thm:special-cases-robust}.
Suppose also some closed convex \emph{robustness set} $\mathcal{R}$ such that $\bm{0} \in \mathcal{R}$ (usually some type of norm-ball).
Then for any $\lv \in \R^{\NGroups}$, the 
Angel's best responses under the following special cases
of Angel action spaces $\mathcal{A}_{\Ang}$, as represented by
weight action spaces $\Wv$, give rise to 
\emph{robust variants} of standard aggregator functions.

\begin{enumerate}
\item\label{thm:special-cases-robust:util} \textbf{Utilitarian}: Suppose $\Wv = (\wv^{*} + \mathcal{R}) \cap \triangle_{\NGroups}$.
Then
$
\displaystyle
\minmax_{\wv \in \Wv} \Payoff_{1} (\lv; \wv) = \minmax_{\wv' \in \Wv} \Mean_{1}(\lv; \wv')
\enspace.
$

\item\label{thm:special-cases-robust:wumswf} \textbf{Weighted Utilitarian-Maximin
}: Suppose $\Wv = (\left\{ \wv \, \middle| \, \wv \succeq \gamma \wv^{*} \right\} + \gamma\mathcal{R}) \cap \triangle_{\NGroups} = \linebreak[3] (\gamma \wv^{*} + \mathcal{R}) \cap \triangle_{\NGroups} + (1 - \gamma) \triangle_{\NGroups}$.
Then
$
\displaystyle
\minmax_{\wv \in \Wv} \Payoff_{1} (\lv; \wv) = \ \ \minmax_{ \mathclap{\wv' \in (\wv^{*} + \mathcal{R}) \cap \triangle_{\NGroups}} } \ \ \gamma \Mean_{1}(\lv; \wv^{*}) + (1 - \gamma) \Mean_{\mp\infty}(\lv)
\enspace.
$

\item\label{thm:special-cases-robust:ggswf} \textbf{Generalized Gini%
}: Suppose
$\Wv = (\left\{ \pi(\wv^{\downarrow}) \middle| \pi \in \Pi_{\NGroups} \right\} + \mathcal{R}) \cap \triangle_{\NGroups}$, 
where $\Pi_{\NGroups}$ is the set of all permutations on $\NGroups$ items.
Then
$
\displaystyle
\minmax_{\wv \in \Wv} \Payoff_{1} (\lv; \wv) = \ \ \minmax_{\mathclap{\smash{\wv^{\downarrow}}' \in (\wv^{\downarrow} + \mathcal{R}) \cap \triangle_{\NGroups}} } \ \ \Mean_{{\wv^{\downarrow}}'}(\lv)
\enspace.
$
\todo{Sorting very unclear!}

\end{enumerate}

Furthermore,\todo{ generally speaking (and in particular for $\mathcal{B}^{\ell_{2}}_{\NGroups}(\gamma)$ with $0 < \gamma \leq ?$ [todo diameter condition]),}
in general none of the above are equivalent to \emph{any} egalitarian, utilitarian, 
weighted utilitarian-maximin, or generalized Gini
welfare or malfare function.
\end{theorem}

\todo{regularized Angel? Modify payoff function?}

\todo{L2 case analysis}

\subsection{From Egocentric to Altruistic Agents}
\label{sec:phil:altruistic-daemon}

\begin{figure}
    \centering

    \begin{minipage}{0.495\textwidth}
\includegraphics[width=\textwidth]{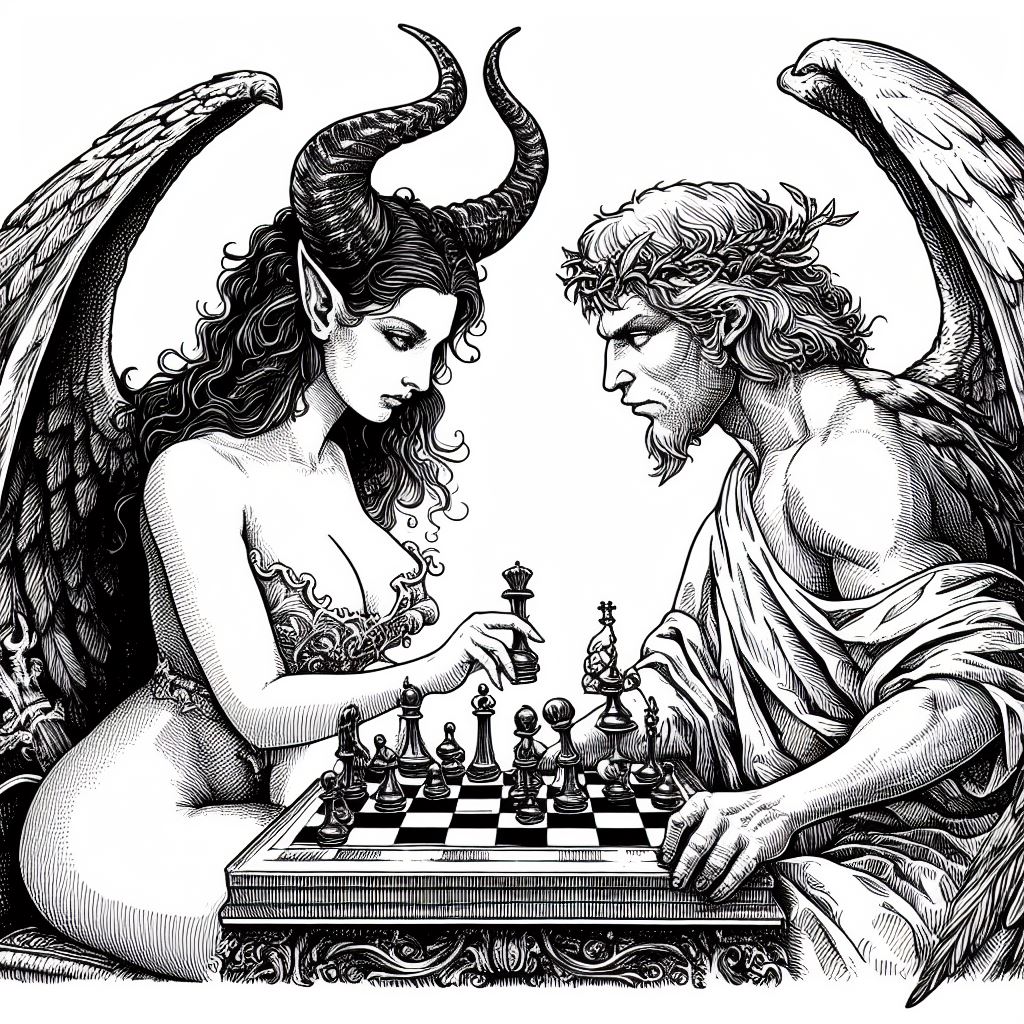}
    \end{minipage}%
    \hfill
    \begin{minipage}{0.495\textwidth}

\ifforc
\footnotesize
\renewcommand{\Large}{\large}
\renewcommand{\medskip}{\smallskip}
\else
\small
\renewcommand{\Large}{\large}
\fi

    \centering 
{ \Large\titlefont The Altruistic D\ae mon Rawlsian Game }

\medskip


$\NGroups, \Lv, \Wv$: Group count, sentiment space, weights space (as in \cref{fig:rawlsian-game,fig:constrained-rawlsian-game}).




$\mathcal{A}_{\Dae} \subseteq \R^{\NGroups}, \mathcal{A}_{\Ang} \subseteq \triangle_{\NGroups}$: Action spaces.

$\Mean(\lv; \wv)$: 
D\ae mon 
aggregator function.

$\Payoff(\lv; \wv): \R^{\NGroups} \times \triangle_{\NGroups} \to \R^{2}$: 
Zero-sum payoff function representing the D\ae mon's aggregate
\vspace{-0.1cm}
\[
\Payoff (\lv, \wv) \doteq \langle \Mean(\lv; \wv), -\Mean(\lv; \wv) \rangle \enspace. 
\]

Strategic gameplay (D\ae mon goes first):
\vspace{-0.1cm}
\[
\argmaxmin_{\lv \in \mathcal{A}_{\Dae}} \minmax_{\wv \in \mathcal{A}_{\Ang}} \Payoff_{1}(\lv, \wv)
= \argmaxmin_{\lv \in \Lv} \minmax_{\wv \in \Wv} \Mean(\lv; \wv) \enspace.
\]

{ \Large\titlefont Or with Utility Transforms}

\medskip

$T(u): \R \to \R$: 
D\ae mon sentiment transform.

$\Payoff (\lv; \wv) \doteq \langle \wv \cdot T(\lv), -\wv \cdot T(\lv) \rangle$:
Payoff function.

Strategic gameplay (D\ae mon goes first):
\vspace{-0.1cm}
\[
\argmaxmin_{\lv \in \mathcal{A}_{\Dae}} \minmax_{\!\wv \in \mathcal{A}_{\Ang}\!} \!\!\! \Payoff_{1}(\lv, \wv)
= \argmaxmin_{\lv \in \Lv} \minmax_{\mathclap{\wv \in \Wv}} T^{-1\!}\left(\wv \cdot T(\lv) \right) \enspace.
\]
    \end{minipage}

    \caption{%
    Metaphoric depiction and game-theoretic description of the 
    altruistic D\ae mon original position game.
    A social-planner D\ae mon (left) plays a zero-sum game against an adversarial Angel (right).
    Both the aggregator-function and the utility-transform formulations of the game are presented. 
%
%
    \todo{Sion equals figure?}
    }
    \label{fig:rawlsian-game-altruistic-daemon}
\end{figure}

\todo{While the prior sections are concerned with adversarial gameplay in 0-sum games, 
[todo: 2 formulations]}

We now show that under our randomized game, if the D\ae mon plays pure strategies and the mixed Angel strategy space is constrained to a compact set, then any power mean welfare function arises as a solution concept when the D\ae mon's payoff is a concave utility transform (or convex disutility transform) of their \emph{ex ante} (dis)utility.
Alternatively, we can think of this as a D\ae mon that is altruistically concerned with the wellbeing of groups of the people in their world, where the Angel is allowed to reweight the sizes of these groups. As a third interpretation, we can think of the game as a metaphysical construct where the D\ae mon is not reincarnated once, but lives all lives within their world, and thus wants to ensure a just and equitable society.

In this game,
the D\ae mon's (dis)utility transform, or their
welfare or malfare function,
determine the power-mean $p$, and the weights $\wv$
and robustness
are determined by the Angel's action space $\Wv$.
Finally, we develop a novel class of
aggregator functions
that combines
the power-mean and the Gini classes, and show that it arises as the solution concept for particular parameterizations of this game.
These games are depicted and described in \cref{fig:rawlsian-game-altruistic-daemon}.

\todo{TODO define payoff function, explain $\Mean_{1}(\lv; \wv) = \wv \cdot \lv$, thus for this choice, it's the same.}

\todo{\paragraph{On Robust Power-Mean Solution Concepts}
[todo pure D\ae mon actions, known utilities?]}

\begin{restatable}[Strategic Gameplay from Nonlinear Objectives]{theorem}{thmstratnonlin}
\label{thm:strat-nonlin}

Suppose payoff function
\\
$\Payoff(\lv; \wv) = \left\langle \Mean_{p}(\lv; \wv), -\Mean_{p}(\lv; \wv) \right\rangle$. 
Then strategic gameplay yields
\[
\argmaxmin_{\lv \in \mathcal{A}_{\Dae}} \minmax_{\wv \in \mathcal{A}_{\Ang}} \!\! \Payoff_{1} (\lv, \wv)
= \argmaxmin_{\lv \in \Lv} \minmax_{\wv \in \Wv} \Mean_{p}( \lv; \wv )
\enspace.
\]
\todo{why arg here?}
Furthermore, if $\Lv = \CH(\Lv)$ and $\Mean(\lv; \wv)$ 
exhibits concave curvature in $\lv$ (or convex for disutility), then a pure D\ae mon strategy is always optimal.
Furthermore, if the curvature is \emph{strictly concave}, then
a pure D\ae mon strategy is \emph{strictly optimal} (over all other pure and mixed D\ae mon strategies).
\end{restatable}

\todo{ similar for any kind of weighted aggregator? Generalized utility transforms?}

\begin{restatable}[Power-Means as 
Utility Transforms]{theorem}{thmpmeanutiltrans}
\label{thm:pmean-util-trans}


Suppose some $p \leq 1$ for utility or $p \geq 1$ for disutility, and a
(dis)utility transform $T(u) = \sgn(p) u^{p}$ 
for $p \neq 0$, or $T(u) = \ln(u)$ for $p=0$, and take $\Payoff_{1} (\lv; \wv) = \wv \cdot T(\lv)$.
Then
\[
\minmax_{\wv \in \mathcal{A}_{\Ang}} \!\! \Payoff_{1} (\lv, \wv)
= \minmax_{\wv \in \Wv} \begin{cases} p \neq 0 & \sgn(p) \Mean_{p}^{p}( \lv; \wv ) \\
    p = 0 & \exp \bigl( \Mean_{0}( \lv; \wv ) \bigr) \enspace.
    \end{cases}
\]
Consequently, as both of the above cases are strict monotonic
functions
of the power-mean $\Mean_{p}(\lv; \wv)$, it holds that
\[
\argmaxmin_{\lv \in \mathcal{A}_{\Dae}} \minmax_{\wv \in \mathcal{A}_{\Ang}} \!\! \Payoff_{1} (\lv, \wv)
= \argmaxmin_{\lv \in \Lv} \minmax_{\wv \in \Wv} \Mean_{p}( \lv; \wv )
  \enspace.
\]

Furthermore, if $\Lv = \CH(\Lv)$ and $T(\cdot)$ 
exhibits concave curvature in $\lv$ (or convex for disutility), then a pure D\ae mon strategy is always optimal.
Furthermore, if the curvature is \emph{strictly concave}, then
a pure D\ae mon strategy is \emph{strictly optimal} (over all other pure and mixed D\ae mon strategies).
\end{restatable}


\paragraph{The $\wv^{\uparrow}$-$p$ Gini Power-Mean Class}
A natural instinct when confronted with the Gini and power-mean classes is to ``combine them'' into something like the following.
\begin{definition}[The $\wv^{\uparrow}$-$p$ Gini Power-Mean Class]
\label{def:gini-pmean}
Suppose some $p \leq 1$ and decreasing weights sequence $\wv^{\downarrow} \in \triangle_{\NGroups}$ for utility, or some $p \geq 1$ and increasing weights sequence $\wv^{\downarrow} \in \triangle_{\NGroups}$ for disutility.
Then, letting $\lv^{\uparrow}$ denote some $\lv \in \RNN^{\NGroups}$ in ascending order,
we define
\begin{align*}
& \Mean_{\wv^{\downarrow}\!,p}(\lv) \doteq \Mean_{p}(\lv^{\uparrow}\!; \wv^{\downarrow}) = \sqrt[p]{ \sum_{i=1}^{\NGroups} \wv^{\downarrow}_{i} (\lv^{\uparrow}_{i})^{p} }
\enspace \text{for 
welfare $(p \leq 1)$} \enspace, 
\enspace \ \  \text{or} \notag \\
& \Mean_{\wv^{\uparrow}\!,p}(\lv) \doteq \Mean_{p}(\lv^{\uparrow}\!; \wv^{\uparrow}) = \sqrt[p]{ \sum_{i=1}^{\NGroups} \wv^{\uparrow}_{i} (\lv^{\uparrow}_{i})^{p} } \enspace \text{for 
malfare $(p \geq 1)$} \enspace,
\end{align*}
i.e., we sort (dis)utilities, 
assign weights in ascending or descending order, and take a weighted power-mean.
\end{definition}

This class clearly generalizes both the unweighted power-mean and Gini families (for $p=1$ and $\wv^{\uparrow} = \frac{1}{\NGroups}\bm{1}$ or $\wv^{\downarrow} = \frac{1}{\NGroups}\bm{1}$), but now combines the piecewise-differentiable 
nature and ordinal boundaries of the Gini family with the continuously-differentiable nonlinear nature of the power-mean family.

Unfortunately, there is no known axiomatic characterization of \cref{def:gini-pmean}, and although the Gini axioms and power-mean axioms overlap heavily, combining them yields only their \emph{intersection}, i.e., the family consisting only of utilitarian and egalitarian welfare or malfare.
However, from either of the above power-mean characterizations (\cref{thm:strat-nonlin,thm:pmean-util-trans}), with the appropriate constrained Angel (selected as in \cref{thm:special-cases-classical}~\cref{thm:special-cases-classical:ggswf}),
\cref{def:gini-pmean} arises as a solution concept to our game.

\todo{\cref{def:gini-pmean} should be referenced later, prove it has curvature, etc.}

\todo{
argue we axiomatically characterize it?
First, it's an alternative Gini, then it generalizes?
D\ae mon accepts power-mean axioms, Angel represents uncertainty, we ``know something about group sizes.''
Not totally satisfying, but it's something.
However, the above 
}

\todo{
Letting $\mathcal{P} \doteq ...$, $\mathcal{G} \doteq ...$, and $\mathcal{GP} \doteq ...$, Clearly $\mathcal{P} \cup \mathcal{G} \subset \mathcal{GP}$, thus ...
}

\subsection{Coercing Altruistic Play from Egocentric D\ae mons}
\label{sec:phil:altruistic-angel}

\begin{figure}
    \centering

    \begin{minipage}{0.495\textwidth}
\includegraphics[width=\textwidth]{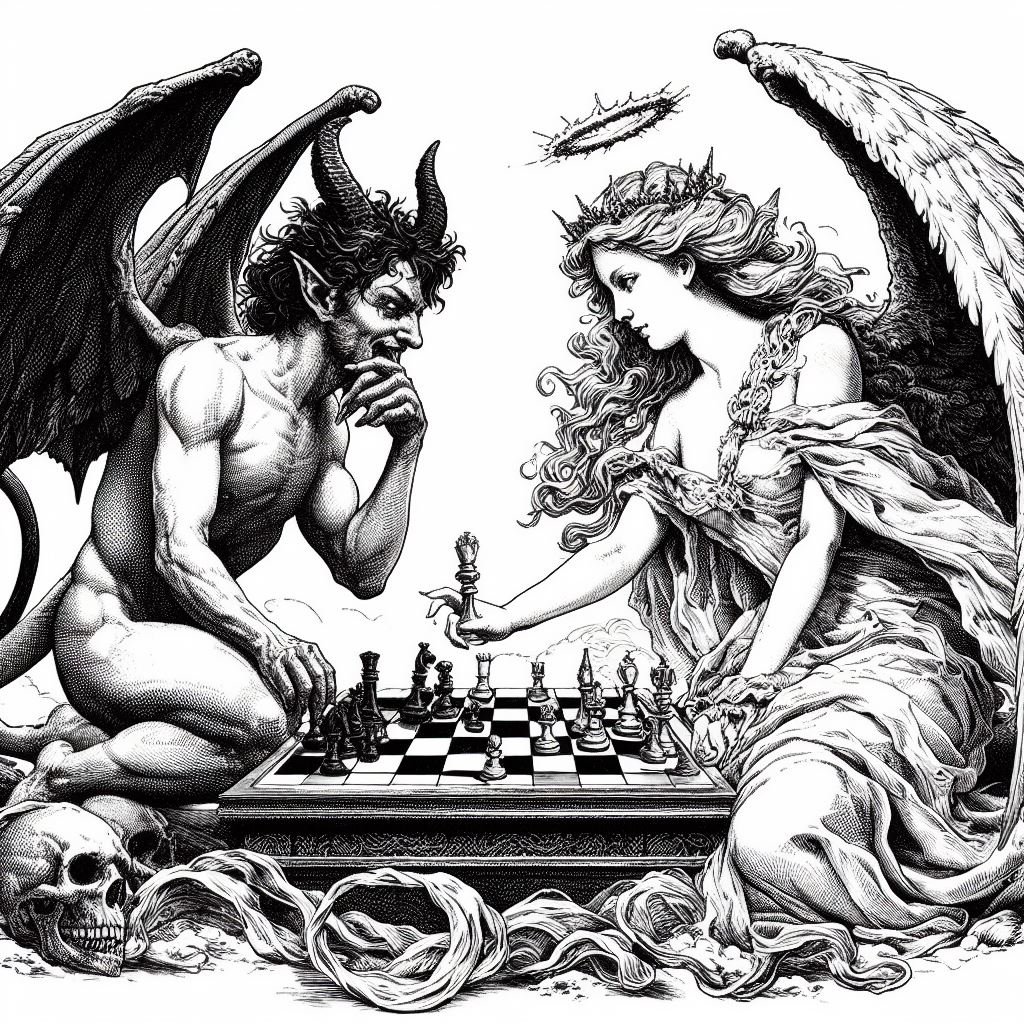}
    \end{minipage}%
    \hfill
    \begin{minipage}{0.495\textwidth}

\ifforc
\footnotesize
\renewcommand{\Large}{\large}
\renewcommand{\medskip}{\smallskip}
\fi

    \centering 
{ \Large\titlefont The Altruistic Angel Rawlsian Game }

\medskip


$\NGroups, \Lv, \Wv$: Group count, sentiment space, weights space 
(as in \cref{fig:rawlsian-game,fig:constrained-rawlsian-game}).




$\mathcal{A}_{\Dae} \subseteq \R^{\NGroups}, \mathcal{A}_{\Ang} = \triangle_{\NGroups}$: Action spaces.

$\Mean(\lv; \wv)$:
Angel 
aggregator function.

$\Payoff(\lv; \wv): \R^{\NGroups} \times \triangle_{\NGroups} \to \R^{2}$: 
Payoff function 
(D\ae mon is self-interested, Angel is altruistic)\todo{leftright langle}%
\vspace{-0.25cm}
\[
\Payoff(\lv, \wv) \doteq \langle \wv \cdot \lv, \Mean(\lv; \wv^{*}) \rangle \enspace. 
\]

Angel action \emph{does not impact} Angel payoff: Any strategy is a ``best response.''

Angel has NFG strategies for which the D\ae mon's best response is to select $\argmaxmin_{\lv \in \Lv} \Mean_{p}(\lv; \wv^{*})$. 

Neither D\ae mon nor Angel has incentive to deviate: This is a Nash equilibrium.

Higher Angel utility is \emph{not possible}, thus this Nash equilibrium is optimal (from Angel's perspective).

    \end{minipage}

    \caption{%
Metaphoric depiction and game-theoretic description of the 
    altruistic Angel original position game.%
    A self-interested D\ae mon (left) 
    is coerced into altruistic play by a social-planner Angel (right).
    }
    \label{fig:enter-label}
\end{figure}

We now show that an altruistic Angel can coerce altruistic play from an egocentric D\ae mon.
Metaphorically, this game is a bit more abstract than those previously discussed, as the D\ae mon still serves as the social planner, but the solution concept we seek optimizes the Angel's aggregator function.

\ifsoft
There is an interesting parallel to representative government here, where the populace (Angel) elects 
leaders (D\ae mon) that perform social planning, but the voting process itself creates incentives for the 
leaders, though we don't see a direct technical connection to our results.
Similarly, we wonder whether altruistic behavior on behalf of corporations, 
such as highly-visible campaigns of corporate ``greenwashing'' or ``rainbow capitalism,'' may arise from similar interactions with consumers.
\fi%

Suppose the Angel has a weighted power-mean aggregator function $\Mean_{p}(\lv; \wv^{*})$.
We construct a payoff function to model the self-centered D\ae mon and altruistic Angel, obtaining%
\begin{equation}
\label{eq:alt-angel-payoff}    
\Payoff(\lv; \wv): \R^{\NGroups} \times \triangle_{\NGroups} \to \R^{2} \doteq \left\langle \wv \cdot \lv, \Mean(\lv; \wv^{*}) \right\rangle \enspace. 
\end{equation}
In this game, the Angel's action space does not represent robustness, but is rather used to influence the actions of the D\ae mon, so we take $\Wv = \triangle_{\NGroups}$.
Until now, we have considered turn-based games, but to analyze Nash equilibria, we must convert the game to normal form.
Here the D\ae mon and Angel act simultaneously, but to preserve the original turn-based game dynamics, the Angel's strategy is conditional on the D\ae mon's action.
In other words, in NFG form, the Angel's strategy space becomes $\Lv \to \triangle_{\NGroups}$, and an Angel strategy $S_{\Ang}(\cdot): \Lv \to \triangle_{\NGroups}$ given any D\ae mon action $\lv \in \Lv$ is to play $S_{\Ang}(\lv)$.

In this game, the Angel action \emph{does not impact} the Angel's payoff, thus any strategy is a ``best response.''
Consequently, the D\ae mon's best response to any Angel strategy is a Nash equilibrium.
The Angel may seem powerless here, but we now show that for a particular choice of Angel strategy, we obtain a Nash equilibrium in which the Angel receives the greatest possible payoff.
\begin{restatable}[Strategic Gameplay in Altruistic Angel Games]{theorem}{thmstrataltangel}
\label{thm:strat-alt-angel}
Suppose the payoff function of \eqref{eq:alt-angel-payoff} for some 
$p > 0$.
If the Angel adopts the strategy $S_{\Ang}(\lv) = \wv_{i} \propto \wv^{*}_{i}\lv_{i}^{p-1}\!$, then the D\ae mon's best response is to select\todo{Explain maxmin?}
\[
\argmaxmin_{\lv \in \Lv} \Payoff_{1}\left(\lv; S_{\Ang}(\lv) \right)
 = \argmaxmin_{\lv \in \Lv} S_{\Ang}(\lv) \cdot \lv
 = \argmaxmin_{\lv \in \Lv} \wv^{*} \cdot \lv^{p}
 = \argmaxmin_{\lv \in \Lv} \Mean_{p}(\lv; \wv^{*})
 \enspace.
\]

Similarly, for $p = 0$, suppose the D\ae mon is limited to utility values at least $\lv_{\min} > 0$, i.e., $\lv \succeq \bm{1} \lv_{\min}$.
If the Angel adopts the strategy $S_{\Ang}(\lv) = \wv_{i} \propto \wv^{*}_{i}\frac{\ln(\lv_{i}/\lv_{\min})}{\lv_{i}/\lv_{\min}}$, then the D\ae mon's best response is to select
\[
\argmax_{\lv \in \Lv} \Payoff_{1}\left(\lv; S_{\Ang}(\lv) \right)
 = \argmax_{\lv \in \Lv} S_{\Ang}(\lv) \cdot \lv
 = \argmax_{\lv \in \Lv} \wv^{*} \cdot \ln \lv
 = \argmax_{\lv \in \Lv} \Mean_{0}(\lv; \wv^{*})
 \enspace.
\]

Finally, for $p < 0$, again suppose the D\ae mon is limited to $\lv \succeq \bm{1} \lv_{\min}$.
If the Angel adopts the strategy $S_{\Ang}(\lv) = \wv_{i} \propto \wv^{*}_{i}
\left(\frac{\lv_{\min}}{\lv_{i}} - (\frac{\lv_{\min}}{\lv_{i}})^{1-p} \right)
$, then the D\ae mon's best response is to select
\[
\argmax_{\lv \in \Lv} \Payoff_{1}\left(\lv; S_{\Ang}(\lv) \right)
 = \argmax_{\lv \in \Lv} S_{\Ang}(\lv) \cdot \lv
 = \argmax_{\lv \in \Lv} 1 - \wv^{*} \cdot \lv^{p}
 = \argmax_{\lv \in \Lv} \Mean_{p}(\lv; \wv^{*})
 \enspace.
\]

Furthermore,
in each case, 
this is a Nash equilibrium, and no strategy profile yields higher Angel payoff (or lower payoff for disutility).
\end{restatable}
\todo{Above is proof, combine cases, explain scaling?!}

This result
tells us that power-mean fairness concepts can arise even from a straightforward linear-utility egocentric D\ae mon.
This is
surprising,
as the results of \cref{sec:phil:altruistic-daemon} need imbue the D\ae mon with a payoff function that already closely 
the power-mean in some way, but
\cref{thm:strat-alt-angel} shows
that we can instead modify the Angel's payoff in a 
sequential
Rawlsian game.
While 
no longer a simple zero-sum normal form game, we still obtain a
Nash equilibrium
in which the power-mean arises as the D\ae mon's robust solution concept.

\todo{Corollary for arbitrary $f$-means? Harmonic welfare?}

\todo{Extend to Gini, robustness.}

\todo{Ancient draft commented.}

\section{Mathematical Properties of Robust Fair Objectives}
\label{sec:math}

We now argue that the objectives of \cref{sec:phil} also arise naturally as robust proxies for unknown information about the relative weights of groups (Angel actions).
In particular, we show theoretical guarantees for the optimization of said robust proxies.
\Cref{sec:math:util} leads with a utilitarian perspective, and \cref{sec:math:prior} generalizes this analysis to a broader class of prioritarian objectives.

\todo{We state results in this section for welfare only for clarity of presentation, because the sense of inequalities would need to flip for malfare.
Unless otherwise noted, all results hold in reverse for disutility and malfare functions. For reference, ``otherwise'' is never noted.\todo{no?}}

We consider two philosophical perspectives on the nature of uncertainty. First, we assume there exists some ground truth weights $\wv^{*}$, but due to epistemic uncertainty about these weights, we 
only have knowledge of some \emph{feasible set} of weights $\Wv$ in which the true weights $\wv^{*}$ are known to be contained.
In the second setting, we do not assume that there exists a single ground truth set of weights, and instead argue that our guarantees hold for any weights vector $\wv \in \Wv$.
The second model feels rather abstract, but is actually quite useful in machine learning contexts:
A model may be trained and then deployed in multiple regions with varying demographics, or demographics may change over time in a single region.
Our 
robust objectives then yield 
model guarantees that
hold so long as demographic shift does not take group frequencies outside of the feasible weight space $\Wv$.

\subsection{A Utilitarian Perspective}
\label{sec:math:util}

Note that, by nature, if $\wv^{*} \in \Wv$, it holds that
\begin{equation}
\label{eq:robust-util}
\inf_{\wv \in \Wv} \wv \cdot \lv \leq \wv^{*} \cdot \lv \leq \sup_{\wv \in \Wv} \wv \cdot \lv \enspace. 
\end{equation}
\todo{
With a probability distribution:
\[
\delta \Mean_{-\infty}(\lv) + (1 - \delta) \inf_{\wv \in \Wv(\delta)} \wv \cdot \lv \leq \Expect_{\wv \distributed \ProbDist_{\Wv}} \left[ \wv \cdot \lv \right] \cdot \lv \leq \delta \Mean_{\infty}(\lv) + (1 - \delta)\sup_{\wv \in \Wv} \wv \cdot \lv \enspace. 
\]
For any ``true objective'' $\Mean^{*}(\cdot; \wv^{*})$, if ... in ..., then the robust objective lower-bounds ...
\[
???
\]
}
Moreover,
this is in some sense the optimal such lower-bound, which holds over adversarial choice of $\wv^{*}$. 

Consequently, optimizing $\argminmax_{\theta \in \Theta} \maxmin_{\wv \in \Wv} \wv \cdot \lv(\theta)$ is a safe proxy for optimizing $\argminmax_{\theta \in \Theta} \wv^{*} \cdot \lv(\theta)$, and the gap between the robust proxy objective and the true objective value can be bounded as
\begin{equation}
\label{eq:robust-util-bounds}
\abs*{ \minmax_{\theta \in \Theta} \wv^{*} \cdot \lv(\theta) - \minmax_{\theta \in \Theta} \maxmin_{\wv \in \Wv} \wv \cdot \lv(\theta) }
\leq
\abs*{ \minmax_{\theta \in \Theta} \minmax_{\wv \in \Wv} \wv \cdot \lv(\theta) - \minmax_{\theta \in \Theta} \maxmin_{\wv \in \Wv} \wv \cdot \lv(\theta) }
\leq
\Range(\lv) \Diam_{1}(\Wv)
\enspace,
\end{equation}
where $\Range(\lv)$ is the sentiment range, and $ \Diam_{1}(\Wv)$ is the $\mathcal{L}_{1}$ diameter of the feasible weights space.
Thus while the Rawlsian game gives us an elegant theoretical model of robust fair objectives, in a practical sense they are also relevant as objectives for operating fairly under adversarial uncertainty, and the magnitude of uncertainty, as measured by $\Diam_{1}(\Wv)$, characterizes the cost of operating under uncertainty via \eqref{eq:robust-util-bounds}.

\todo{Uncertainty sets, upper/lower bounds? say something?}

\subsection{A Generalized Prioritarian Perspective}
\label{sec:math:prior}

The analysis of \cref{sec:math:util} considers \emph{robustness} in the sense of adversarial uncertainty over the weights $\wv$, but not \emph{fairness}, in the sense of nonlinear aggregator functions that incentivize equitable redistribution of (dis)utility.
We now show nonlinear variants of the above results for objectives that can be decomposed as\todo{connect to new notation}
\[
\min_{\wv \in \Wv} \Mean(\lv; \wv)
\enspace. \todo{or max. inf or sup?}
\]
Note that here $\Wv$ is not necessarily the Angel's action space, but rather it is only the robustness parameters that are not incorporated into the objective function itself (see~\cref{thm:special-cases-robust}).
In terms of usage, we may generally assume that a fair objective $\Mean(\cdot; \wv)$ is known (selected) in advance, and epistemic uncertainty about weights is also known and provided
as $\Wv$. 
We thus have
\todo{keep going here.}
\todo{abrupt equation.}
\begin{equation}
\label{eq:robust-welfare}    
\inf_{\wv \in \Wv} \Mean(\lv; \wv) \leq \Mean(\lv; \wv^{*}) \leq \sup_{\wv \in \Wv} \Mean(\lv; \wv) \enspace.
\end{equation}

\todo{Corresponding result to \eqref{eq:robust-util-bounds}.}

\section{Adversarial Optimization of Robust Fair Objectives}
\label{sec:adv-opt}

This work
centers the motivation for and properties of robust fair objectives, but we now briefly discuss
their optimization.
We then discuss modeling and applications in fair allocation and machine learning.

In this section, we assume an objective of the form\todo{number and ref eq. also connect to \cref{def:gini-pmean}.}
\[
\argmaxmin_{\lv \in \Lv} \minmax_{\wv \in \Wv} \Mean_{p}( \lv; \wv )    \enspace,
\]\todo{Introduce shorthand notation here!}
where either $\Mean(\lv; \wv)$ exhibits concavity in $\lv$ and convexity and $\wv$ for outer maximization, or convexity in $\lv$ and concavity in $\wv$ for outer minimization.
It is well known from convex optimization theory that we can efficiently maximize concave functions or minimize convex functions. Moreover, inner maximization preserves outer concavity and inner minimization preserves outer convexity.
Thus the cards seem to be in our favor, 
and adversarial optimization exhibiting this concave-convex maximization convex-concave minimization structure is generally tractable \citep{nemirovski2004prox,lin2020near}, e.g., via a variety of gradient ascent-descent methods.

\begin{lemma}[Power-Mean Curvature]
\label{lemma:pmc}
Power-mean welfare and malfare functions exhibit the following curvature.
\begin{enumerate}
\item For any $p \geq 1$, $\Mean_{p}(\cdot; \wv)$ is convex, (strictly, but never strongly, for $p > 1$), and $\Mean_{p}(\lv; \cdot)$ is concave (non-strictly).

\item For any $p \leq 1$, $\Mean_{p}(\lv; \cdot)$ is concave, (non-strictly), and $\Mean_{p}(\lv; \cdot)$ is convex (non-strictly).

\todo{sup / inf preserves curvature}
\end{enumerate}
\end{lemma}

In other words, the power-mean $\Mean_{p}(\lv; \wv)$ exhibits opposite curvature in $\lv$ and $\wv$.
\todo{
NB strict and/or strong convexity and/or concavity may arise if $\lv$ is bounded away from $0$ and/or bounded above, and/or if $\wv$ is bounded away from $0$. 
\\
TODO cite advanced adversarial optimization methods,
discuss smoothness problems.
}%
We will first consider a few trivial cases, where \cref{lemma:pmc} suffices, as it is easy to convert the space of feasible allocation to the space of feasible utility values.
In such settings, it is straightforward to apply standard maximin-optimization algorithms.

In general, it is not always so easy to convert between the parameter space space of feasible allocations $\Theta$ and the space of feasible utilities.
In these more general settings, we need to consider the optimization problem directly as a function of the parameters space $\Theta$.
We thus require another technical lemma.

\begin{lemma}[Power-Mean Composition Curvature]
\label{lemma:pmc-comp}
Suppose some
per-group (dis)utility
function $\lv: \Theta \to \RNN^{\NGroups}$.
Compositions $\Mean_{p}(\lv(\theta); \wv): (\Theta \times \triangle_{\NGroups}) \to \RNN$ of power-means
with $\lv$ exhibit the following curvature.
\begin{enumerate}
\item For any $p \geq 1$, if $\lv: \Theta \to \R^{\NGroups}$ is convex, then $\Mean_{p}(\lv(\theta); \wv): (\Theta \times \triangle_{\NGroups}) \to \RNN$ is convex in $\theta$ and concave in $\wv$.

\item For any $p \leq 1$, if $\lv: \Theta \to \R^{\NGroups}$ is concave, then $\Mean_{p}(\lv(\theta); \wv): (\Theta \times \triangle_{\NGroups}) \to \RNN$ is concave in $\theta$ and convex in $\wv$.
\end{enumerate}

\todo{Robust ops}
\end{lemma}

\subsection{Simple Applications in Fair Allocation Problems}

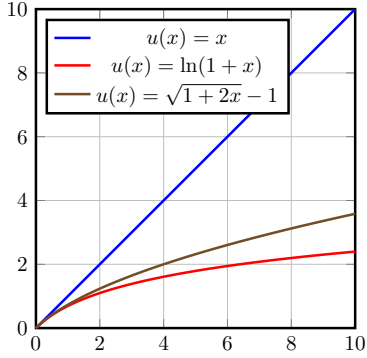
\begin{figure}
    \centering

\begin{minipage}{0.349\textwidth}

\centering

\scalebox{0.8}{%
\begin{tikzpicture}

\ifforc
\newcommand{\figsize}{1.01\textwidth}\hspace{-0.2cm}
\else
\newcommand{\figsize}{0.88\textwidth}
\fi
\begin{axis}[
    scale only axis=true,
    width=\figsize,
    height=\figsize,
    domain=0:10,
    xmin=0,xmax=10,
    ymin=0,ymax=10,
    samples=64,
    smooth,very thick,
    no markers,
    grid,
    legend pos=north west,
  ]

\plot {x};
\addlegendentry{$u(x) = x$}

\plot {ln(1 + x)};
\addlegendentry{$u(x) = \ln(1 + x)$}

\plot {sqrt(1 + 2 * x) - 1};
\addlegendentry{$u(x) = \sqrt{1 + 2x} - 1$}

\end{axis}

\end{tikzpicture}%
}

\end{minipage}%
\hfill%
\begin{minipage}{0.649\textwidth}

\caption{Plots of various nonlinear utility transform functions referenced in the text, as a compared to linear utility $u(x) = x$. The logarithmic transform $u(x) = \ln(1 + x)$ and square root transform $u(x) = \sqrt{1 + 2x} - 1$ on utility values are shown.
Note that both are smooth strictly-increasing strictly concave utility transforms that are tangent to the linear utility $u(x) = x$ at $x = 0$, thus they obey $0 \leq u(x) \leq x$, $\lim_{x \to \infty} u(x) = \infty$, $\lim_{x \to \infty} \frac{u(x)}{x} = 0$, and $\lim_{x \to 0^{+}} \frac{u(x)}{x} = 1$, i.e., they lie strictly below linear utility, and behave asymptotically as 
sublinear but superconstant (unbounded).%
\todo{Also do the power-mean transforms in a subfigure? Range 0--2?}%
\todo{Legend on side?}
}
\label{fig:concave-utility}

\end{minipage}
\end{figure}

We assume here that
$\NGroups$ agents are being allocated
$k$ divisible goods.
Each good $i$ has
\emph{capacity} $\bm{c}_{i}$, thus allocations are \emph{matrices}\todo{bm theta?} $\theta \in \RNN^{\NGroups \times k}$,
where
each column (good allocation) $i$
sum is bounded by
$\bm{c}_{i}$.
We let $\Theta$ denote the set of \emph{feasible allocations}, and $\lv(\theta) \in \RNN^{\NGroups}$ is the \emph{utility vector} given some feasible allocation $\theta \in \Theta$.
We thus have
\[
\theta \ = \, \begin{array}{r}
\rotatebox{45}{\text{\small\!\! Good $1$\!}} \quad \quad \ \ \rotatebox{45}{\text{\small\!\! Good $k$\!} \ \ \  } \!\!\!\!\! \null \\
\begin{blockarray}{c(ccc)}
\text{\small Agent 1} & \theta_{1,1} & \cdots & \theta_{1,k} \\
 & \vdots & \ddots & \vdots \\
\text{\small Agent $\NGroups$} & \theta_{\NGroups,1} & \cdots & \theta_{\NGroups,k} \\
\end{blockarray} \enspace.
\end{array}
\]
\todo{Transposition is:
\[
\theta \ = \, \begin{array}{r}
\rotatebox{45}{\text{\small\!\! Agent 1}} \quad \quad \ \rotatebox{45}{\text{\small\!\! Agent $\NGroups$} \ \ \, } \!\!\!\!\!\!\! \null \\
\begin{blockarray}{c(ccc)}
\text{\small Good 1} & \theta_{1,1} & \cdots & \theta_{1,\NGroups} \\
 & \vdots & \ddots & \vdots \\
\text{\small Good $k$} & \theta_{k,1} & \cdots & \theta_{k,\NGroups} \\
\end{blockarray} \enspace.
\end{array}
\]
}%
We now study the optimization problems that result from applying robust fair objectives to such tasks.

\begin{example}[Single Good with Linear Utility]
Suppose we have $k=1$ goods with capacity $c$, and linear utility $\bm{p}_{i}$ for agent $i$ in their share, i.e., $\lv_{i}(\theta) = \theta_{i,1} \bm{p}_{i}
$.
Then
\begin{align*}
\max_{\theta \in \Theta} \min_{\wv  \in \Wv} \Mean(\lv(\theta); \wv)
    &= \max_{\substack{\theta \in \RNN^{\NGroups} : \\ \sum_{i=1}^{\NGroups} \theta_{i,1} \leq c }} \min_{\wv \in \Wv} \Mean( i \mapsto \bm{p}_{i} \theta_{i,1}; \wv) & \\
    &= \max_{\substack{\lv \in \RNN^{\NGroups} : \\ \sum_{i=1}^{\NGroups} \frac{1}{\bm{p}_{i}} \theta_{i,1} \leq c }} \min_{\wv \in \Wv} \Mean(\lv; \wv) \enspace.
\end{align*}
We may approximately solve the convex-concave optimization problem of the RHS via standard maximin programming techniques, and
given a maximal $\lv$ in the RHS objective, we can trivially identify some $\theta \in \Theta$ that gives rise to it in the LHS via 
linear programming, 
or in closed form by 
inversion to get $\theta_{i,1} = \frac{\lv_{i}(\theta)}{\bm{p}_{i}}$.
\end{example}

\begin{example}[Single Good with Nonlinear Utility]
Suppose \emph{nonlinear utility} $\lv_{i}(\theta) = (\sqrt{1 + 2\theta_{i,1}} - 1) \bm{p}_{i}$ (see~\cref{fig:concave-utility}).
With this very special choice, we can similarly pose
\begin{align*}
\max_{\theta \in \Theta} \min_{\wv \in \Wv} \Mean(\lv(\theta); \wv)
    &= \max_{\substack{\theta \in \RNN^{\NGroups} : \\ \sum_{i=1}^{\NGroups} \theta_{i,1} \leq c }} \min_{\wv \in \Wv} \Mean\left( i \mapsto (\sqrt{1 + 2\theta_{i,1}} - 1) \bm{p}_{i}; \wv\right) & \\
    &= \max_{\substack{\lv \in \RNN^{\NGroups} : \\ \frac{1}{2}\sum_{i=1}^{\NGroups} (\frac{\lv_{i}}{\bm{p}_{i}} + 1)^2 - 1  \leq c }} \min_{\wv \in \Wv} \Mean(\lv; \wv) \enspace.
\end{align*}
Again we may approximately solve the convex-concave optimization problem of the RHS via standard maximin programming techniques.
The feasible set of utility values in the RHS is the nonnegative 
portion of an axis-aligned ellipsoid, and converting some optimal $\lv$ to the $\theta$ that gives rise to it is trivial via 
convex quadratic programming, 
or via inversion through the quadratic formula to get $\theta_{i,1} = \frac{1}{2}(\frac{\lv_{i}}{\bm{p}_{i}} + 1)^{2} - \frac{1}{2} = \frac{\lv_{i}}{\bm{p}_{i}} + \frac{\lv_{i}^{2}}{2\bm{p}_{i}^{2}}$.
\end{example}
\begin{example}[
Multiple Goods with Linear Utility]
Now suppose $k$ goods with \emph{linear utility}
$\lv_{i}(\theta) = \theta_{i} \cdot \bm{P}_{i} = \sum_{j=1}^{k} \theta_{i,j} \bm{P}_{i,j}$. Suppose also arbitrary linear equality and inequality constraints on $\Theta$, with $\Theta \neq \emptyset$.
Via similar techniques, we can convert the space of feasible $\theta \in \Theta$ to some $\lv \in \Lv$, where both $\Theta$ and $\Lv$ are polytopes.
We can then optimize over $\lv \in \Lv$, and finally select some $\theta \in \Theta$ that gives rise to the optimal $\lv$ via 
linear programming. 
\todo{display math}
\end{example}

\todo{New subsection around here?}

In each of these
examples,
we optimize a robust fair objective over utility values, and then
invert
to obtain an allocation $\theta \in \Theta$ that produces utility values that optimize the robust objective.
%
%
%
In many applications,
the feasible space of (dis)utility vectors $\Lv$ is not directly known.
Instead, we now assume that
we have
some \emph{parameter} $\theta \in \Theta$,
which yields a utility vector $\lv(\theta)$.
Robust fair objectives of $\lv(\theta)$ can then be optimized.
This is directly relevant to many fair ML applications, but we note now that fair allocation of divisible goods (
or chores) with nonlinear utility can also be handled in this way, so long as utilities are \emph{concave} in $\theta$ and disutilities are \emph{convex} in $\theta$.

For example, suppose store (agent) $i$ will sell up to $\bm{C}_{i,j}$ units of items $j$ for $\bm{P}_{i,j} \$$/unit of profit, and the \emph{utility} derived by the store is given by the logarithmic utility transform $\ln(1 + \theta_{i} \cdot \bm{P}_{i})$ (see~\cref{fig:concave-utility}). 
Then, for $\NGroups$ stores and $k$ items, the utility vector of the allocation $\theta \in \RNN^{\NGroups \times k}$ can be expressed as
\[
\lv_{i}(\theta)  = \ln\left(1 + \vphantom{\frac{1}{2}}\smash{\sum_{i=1}^{k}} \bm{P}_{i,j}\min(\bm{C}_{i,j}, \theta_{i,j}) \right) \enspace.
\]

The final task is then to compute
\[
\argmax_{\theta \in \Theta} \min_{\wv \in \Wv} \Mean\bigl( \lv(\theta) ; \wv \bigr) \enspace,
\]
which by \cref{lemma:pmc} is
tractable for power-mean welfare objectives. 


This idea immediately extends to ML settings, where the (empirical) loss or utility derived by each agent or group is also a complicated function of some parameter $\theta \in \Theta$.
\ifextensions
The optimization of such objectives is straightforward via standard first-order methods, e.g., subgradient descent for Lipschitz continuous malfare functions as proposed by \citet{cousins2021axiomatic}.
\Cref{sec:stat-fml} shows that robust power-mean malfare functions are indeed Lipschitz continuous, and 
then derives generalization bounds for fair learning with such robust objectives.
\fi

\todo{optimization problems draft.}

\ifextensions

\section{Statistical Generalization Bounds for Robust Fair Objectives}
\label{sec:stat-fml}

We now show that
our robust objectives preserve the Lipschitz or H\"older continuity of their underlying non-robust counterparts.
In particular, robust power-mean malfare functions are Lipschitz continuous, and robust power-mean welfare functions are almost always Lipschitz or H\"older continuous. 
We then translate these results into generalization bounds for fair learning with such robust objectives.

It is known 
\citep{cousins2022uncertainty}
that for any monotonic aggregator function $\Mean(\cdot)$, if the gap between true and empirical risk or utility values $\lv_{i}$ and $\hat{\lv}_{i}$ of each group $i$ is 
no more than $\epsv_{i}$, 
i.e., if we have $\abs{ \lv_{i} - \hat{\lv}_{i} } \leq \epsv_{i}$, then
we may bound the generalization error as
\begin{equation}
\label{eq:genbound}
\Mean(\hat{\lv} - \epsv) \leq \Mean(\lv) \leq \Mean(\hat{\lv} - \epsv)
\enspace.
\end{equation}
Bounding the estimation error $\epsv_{i}$ for each group is a nontrivial matter, but standard techniques in statistical learning theory suffice.
\citet{cousins2023revisiting} uses Rademacher averages and other statistical methods to bound the supremum deviation over each group, thus deriving values for $\epsv$, and \citet{cousins2024pool} show that sharper generalization bounds can be achieved by explicitly considering the effect of fair training on generalization error, and in particular that from the perspective of each group i, the fair model is effectively learned over a class that is biased towards strong performance on the remaining groups.

From \eqref{eq:genbound}, Lipschitz or H\"older continuity properties of the aggregator function
yield
loose but algebraically convenient bounds, as well as bounds on the sample complexity of PAC-learning such objectives.
In particular, a function $\Mean(\cdot)$ is $\lambda$-$\alpha$-$\norm{\cdot}_{\Mean}$ H\"older continuous \wrt\ some norm $\norm{\cdot}_{\Mean}$ if for all $\lv, \lv'$ in its domain, it holds
\begin{equation}
\label{eq:hcont}
\abs{\Mean(\lv) - \Mean(\lv')} \leq \lambda \norm{\lv - \lv'}_{\Mean}^{\alpha} \enspace.
\end{equation}
Moreover, if $\alpha = 1$, then $\Mean(\cdot)$ is $\lambda$-$\norm{\cdot}_{\Mean}$ Lipschitz continuous.
In concert with \cref{eq:genbound}, under H\"older continuity, we may thus bound generalization error as\todo{talk about uniform convergence?}
\begin{equation}
\label{eq:hgenbound}
\abs{ \Mean(\hat{\lv} ) - \Mean(\lv) } \leq \lambda \norm{\epsv}_{\Mean}^{\alpha}
\enspace,
\end{equation}
and plugging in a specific expression for per-group generalization error $\epsv$ allows us to solve for sample complexity (sufficient sample size) bounds.
Using these results, we need only show that fair robust objectives exhibit similar Lipschitz and H\"older continuity properties to their non-robust counterparts.

\subsection{Lipschitz and H\"older Continuity of Robust Fair Objectives}
\label{sec:stat-fml:cont}


To streamline the analysis of our robust fair objectives, we introduce the notation
\[
\Mean(\lv; \Wv) \doteq \supinf_{\wv \in \Wv} \Mean(\lv; \wv)
\enspace.
\]
We now show that these objectives have similar continuity properties to their non-robust counterparts, in line with lemmata 3.12 and 3.13 of \citet{cousins2023revisiting}.

\begin{restatable}[H\"older Continuity of Robust Fair Objectives]{lemma}{lemmarobustholder}
\label{lemma:robust-holder}
Suppose an $\lambda$-$\alpha$-$\norm{\cdot}_{\Mean}$ H\"older continuous weighted aggregator function $\Mean(\lv; \wv)$ over feasible weights space $\Wv \subseteq \triangle_{\NGroups}$.
Then the robust aggregator function $\Mean(\lv; \Wv) = \maxmin_{\wv \in \Wv} \Mean(\lv; \wv)$ is $\lambda$-$\alpha$-$\norm{\cdot}_{\Mean}$ H\"older continuous,
 i.e., for all $\lv,\lv'$, it holds
\[
\abs{\Mean(\lv; \Wv) - \Mean(\lv'; \Wv)} \leq \lambda \norm{\lv - \lv'}_{\Mean}^{\alpha}
\enspace.
\]
%
\end{restatable}

\begin{restatable}[H\"older Continuity of Robust Power-Means]{corollary}{coropmeanholder}
\label{coro:pmean-holder}
Suppose a robust power-mean operator $\Mean_{p}(\lv; \Wv)$,
sentiment
range $\frange$,
and 
arbitrary $\lv,\lv' \in [0, \frange]^{\NGroups}$.
Let \linebreak[3]$\wv_{\min} \doteq \inf_{\wv \in \Wv} \min_{i \in 1,\dots,\NGroups} \wv_{i}$, and $\wv_{\max} \doteq \sup_{\wv \in \Wv} \max_{i \in 1,\dots,\NGroups} \wv_{i}$.
Then $\Mean_{p}(\lv; \Wv)$ exhibits the following Lipschitz and H\"older continuity properties.
\begin{enumerate}
\item
For all $p \geq 1$: $\Mean_{p}(\lv; \Wv)$ is $1$ Lipschitz continuous \wrt\ itself, i.e.,
\[
\abs{ \Mean_{p}(\lv; \Wv) - \Mean_{p}(\lv'; \Wv) } \leq \Mean_{p}( \abs{\lv - \lv'}; \Wv) \leq
\norm{\lv - \lv'}_{\infty}
\enspace.
\]
Thus 
$\Mean_{1}(\lv; \Wv)$ is 
$\wv_{\max}$-$\norm{\cdot}_{1}$ Lipschitz continuous, and for $p = \infty$, 
it is $1$-$\norm{\cdot}_{\infty}$ Lispchitz continuous, and both of these constants are optimal.\todo{skip infty?}

\item For all $p < 0$, if $\wv_{\min} > 0$, then 
$\Mean_{p}(\lv; \Wv)$ is
$\frac{1}{\sqrt[\abs{p}]{\wv_{\min}}}$%
-$\norm{\cdot}_{\infty}$ Lipschitz continuous.

\item For all $p \in (0, 1)$,
$\Mean_{p}(\lv; \Wv)$ is $\frange^{1 - p}\frac{1}{p}$-$p$-$\norm{\cdot}_{\infty}$ H\"older continuous.\todo{Other norm cases.}

\item For all $p \leq 1$, 
if $\wv_{\min} > 0$, then
$\Mean_{p}(\lv; \Wv)$ is $\frange^{1 - \wv_{\min}}$-$\wv_{\min}$-$\norm{\cdot}_{\infty}$ H\"older continuous.

\end{enumerate}
\end{restatable}

\todo{talk about gini pmean class? what is wmin/wmax then?}

\todo{gini coro?
\begin{restatable}[Lipschitz Continuity of Gini Aggregator Functions]{corollary}{corolipgini}
\label{coro:lip-gini}
???
\end{restatable}
}

\todo{canibalize logarithmic mean?}
\todo{averaging over Angel space, or taking any weighted average or pmean malfare, also preserves H\"older?}

\todo{Lipschitz optimization theorem? Ref axiomatic paper?}

\todo{Explain why we care.}

\begin{restatable}[A Codex of Sample Complexity]{theorem}{thmcodexsc}
\label{thm:codex-sc}

Suppose that $\Mean(\lv)$ is $\lambda$-$\alpha$-$\norm{\cdot}_{\Mean}$ H\"older continuous, and for any \emph{sample size} $m \geq m_{0}$ and \emph{failure probability} $\delta \in (0, 1)$, the empirical and true per-group risk values $\hat{\lv}$ and $\lv$ of the empirical malfare minimizer
on $m$ samples obey
$\Prob \left( \abs{ \hat{\lv_{i}} - \lv_{i} } > \sqrt{\frac{\bm{v}_{i}\ln \frac{t}{\delta}}{m}} \right) < \delta$\todo{explain v, t here}.
Then the \emph{sample complexity} of empirical malfare minimization
(or welfare maximization)
obeys
\[
m^{*}(\varepsilon, \delta) \leq \max\left( m_{0}, \ceil*{\left(\frac{\lambda}{\varepsilon}\right)^{2/\alpha} \norm{ i \mapsto  \sqrt{\bm{v}_{i}} }_{\Mean}^{2} \ln \frac{t\NGroups}{\delta}} \right)
\enspace,
\]
i.e., for any sample size $m \geq m^{*}(\varepsilon, \delta)$, it holds that
\[
\Prob\left( \Mean(\lv) \leq \Mean(\lv^{*}) + \varepsilon \right) \geq 1 - \delta
\enspace \text{for malfare}
\enspace,
\quad
\text{or} \enspace
\Prob\left( \Mean(\lv) \geq \Mean(\lv^{*}) - \varepsilon \right) \geq 1 - \delta
\enspace \text{for welfare}
\enspace.
\]
\todo{W/M order?}
\todo{prob over training data and thus $\hat{h}$ and empirical risk}
where
$\lv$ is the true risk (utility) vector of the learned model, and $\Mean(\lv^{*})$ is the infimum of true malfare (welfare) over the hypothesis class.
\todo{Generalize to mixed rates?}
\cyrus
{Proof sketch:
Union bound into
\[
\varepsilon
= \lambda \norm*{ i \mapsto  \sqrt{\frac{\bm{v}_{i}\ln \frac{t\NGroups}{\delta}}{m}} }_{\Mean}^{\alpha}
= \lambda \norm{ i \mapsto  \sqrt{\bm{v}_{i}} }_{\Mean}^{\alpha} \left( \frac{\ln \frac{t\NGroups}{\delta}}{m} \right)^{\alpha/2}
\implies m = \left(\frac{\lambda}{\varepsilon}\right)^{2/\alpha} \norm{ i \mapsto  \sqrt{\bm{v}_{i}} }_{\Mean}^{2} \ln \frac{t\NGroups}{\delta}
\]
samples suffice.
}
\end{restatable}

\todo{2 factor stuff?}

From \cref{thm:codex-sc}, it immediately follows that if each group's generalization error behaves as $\sqrt{\frac{v\ln \frac{t}{\delta}}{m}}$ (i.e., variance proxy $v$, tail count $t$; such bounds arise frequently via Rademacher averages and other statistical techniques \citep{cousins2022uncertainty,cousins2024pool}), the sample complexity of optimizing any robust power-mean malfare function is bounded by\todo{2v?}
\[
m^{*}(\varepsilon, \delta) \leq
\ceil*{ \frac{v\lambda^{2}\ln \frac{t\NGroups}{\delta}}{\varepsilon^{2}} }
\enspace.
\]

\fi

\section{Conclusion}
\label{sec:conc}

We find that fair robust 
learning and optimization tasks
can be expressed as maximin or minimax optimization problems, and efficiently solved via standard convex optimization methodology.
In continuous optimization settings, as arise in machine learning and fair allocation with divisible goods, standard techniques from the adversarial optimization literature are appropriate.
Furthermore, the inner maximization is often representable in closed form, and is thus amenable to general optimization settings, e.g., in sequential allocation \citep{viswanathan2023general,cousins2023dividing,cousins2023good
} tasks.

\todo{emph pmeans / others arise. util to egal is obvious, others are surprising?}

\todo{: ADDRESS THESE
- The conceptual contribution of the paper is somewhat limited. For example, the relationship between the fairness definitions considered in the paper and robustness is perhaps less surprising, because such a min-max perspective is fundamental to maximin fairness.
- The technical contribution of the paper is also a bit limited. For example, the results in Section 6 about the optimizing these fairness objectives seem to follow directly from standard optimization tools, and the results in Section 7 seem to similarly follow from known results. Moreover, although the paper includes several theorems, many of these results seem to follow in a fairly straightforward manner from definition.
- The presentation of the paper is difficult to follow. For example, I found it a bit difficult to extract the key points of the paper.
}

\ifextensions
We show that the concept of robust fair objectives that arise from our constrained Rawlsian games preserve the Lipschitz and/or H\"older continuity properties of their underlying welfare or malfare concepts.
This has implications to their optimization via convex optimization convergence rates, and is also directly relevant to the generalization error and sample complexity of fair learning these objectives.
\fi

We stress that while uncertainty about the world, in particular about the weights $\wv$ (or relative sizes) of each group, does give rise to 
many standard fairness concepts,
\todo{ l1 is umswf }%
these fairness concepts are also inherently motivated through the lens of fairness itself.
In other words, we help needy groups because we feel that it is the right thing to do, not because we labor under the delusion that we may someday \emph{literally become them}.
It is a deep philosophical 
question of human nature as to whether our capacity for empathy inherently predisposes us to think in this way, though
such philosophical quandaries are
well beyond the scope of
this
work.

{
\def\bibfont{\normalfont}
\setlength{\bibsep}{4pt}
%
%
%
%
%
%
\bibliographystyle{plainnat}
\bibliography{bibliography,cyrus}
}

\end{document}